\pdfoutput=1
\RequirePackage{ifpdf}
\ifpdf % We~are running pdfTeX in pdf mode
\documentclass[pdftex]{sigma}%,draft
\else
\documentclass{sigma}
\fi

\numberwithin{equation}{section}

\newtheorem{Theorem}{Theorem}[section]
\newtheorem{Lemma}[Theorem]{Lemma}
\newtheorem{Proposition}[Theorem]{Proposition}
 { \theoremstyle{definition}
\newtheorem{Definition}[Theorem]{Definition}
\newtheorem{question}[Theorem]{Question}
\newtheorem{Example}[Theorem]{Example}}

\newcommand{\dd}{{\rm d}}
\newcommand{\TT}[1]{\begin{matrix}T\vspace{-6pt}\\{\scriptscriptstyle #1}\vspace{-2pt} \\T\end{matrix}}
\newcommand{\TA}[1]{\begin{matrix}T\vspace{-6pt}\\{\scriptscriptstyle #1}\vspace{-2pt} \\A\end{matrix}}
\newcommand{\AT}[1]{\begin{matrix}A\vspace{-6pt}\\{\scriptscriptstyle #1}\vspace{-2pt} \\T\end{matrix}}
\renewcommand{\AA}[1]{\begin{matrix}A\vspace{-6pt}\\{\scriptscriptstyle #1}\vspace{-2pt} \\A\end{matrix}}
\newcommand{\FF}{{\mathbb F}}
\newcommand{\ZZ}{{\mathbb Z}}

\begin{document}
\allowdisplaybreaks

\newcommand{\arXivNumber}{2102.12383}

\renewcommand{\thefootnote}{}

\renewcommand{\PaperNumber}{100}

\FirstPageHeading

\ShortArticleName{$c_2$ Invariants of Hourglass Chains via Quadratic Denominator Reduction}

\ArticleName{$\boldsymbol{c_2}$ Invariants of Hourglass Chains\\ via Quadratic Denominator Reduction\footnote{This paper is a~contribution to the Special Issue on Algebraic Structures in Perturbative Quantum Field Theory in honor of Dirk Kreimer for his 60th birthday.
The~full collection is available at \href{https://www.emis.de/journals/SIGMA/Kreimer.html}{https://www.emis.de/journals/SIGMA/Kreimer.html}}}

\Author{Oliver SCHNETZ~$^{\rm a}$ and Karen YEATS~$^{\rm b}$}

\AuthorNameForHeading{O.~Schnetz and K.~Yeats}

\Address{$^{\rm a)}$~Department Mathematik, Friedrich-Alexander-Universit\"at Erlangen-N\"urnberg,\\
\hphantom{$^{\rm a)}$}~Cauerstrasse 11, 91058, Erlangen, Germany}
\EmailD{\href{mailto:schnetz@mi.uni-erlangen.de}{schnetz@mi.uni-erlangen.de}}
\URLaddressD{\url{https://www.math.fau.de/person/oliver-schnetz/}}

\Address{$^{\rm b)}$~Department of Combinatorics and Optimization, University of Waterloo,\\
\hphantom{$^{\rm b)}$}~Waterloo, Ontario, N2L 3G1, Canada}
\EmailD{\href{mailto:kayeats@uwaterloo.ca}{kayeats@uwaterloo.ca}}

\ArticleDates{Received February 25, 2021, in final form November 02, 2021; Published online November 10, 2021}

\Abstract{We introduce families of four-regular graphs consisting of chains of hourglasses which are attached to a finite kernel. We prove a formula for the $c_2$ invariant of these hourglass chains which only depends on the kernel. For different kernels these hourglass chains typically give rise to different $c_2$ invariants. An exhaustive search for the $c_2$ invariants of hourglass chains with kernels that have a maximum of ten vertices provides Calabi--Yau manifolds with point-counts which match the Fourier coefficients of modular forms whose weights and levels are [4,8], [4,16], [6,4], and [9,4]. Assuming the completion conjecture, we show that no modular form of weight~2 and level $\leq1000$ corresponds to the $c_2$ of such hourglass chains. This provides further evidence in favour of the conjecture that curves are absent in~$c_2$ invariants of $\phi^4$ quantum field theory.}


\Classification{81T18}

\renewcommand{\thefootnote}{\arabic{footnote}}
\setcounter{footnote}{0}

\section{Introduction}

Given a graph $G$, to each edge $e\in E(G)$ associate a variable $\alpha_e$ and define the \emph{Kirchhoff polynomial} or \emph{first Symanzik polynomial} to be
\begin{gather*}
\Psi_G = \sum_{T}\prod_{e\not\in T}\alpha_e,
\end{gather*}
where the sum is over all spanning trees $T$ of $G$. When it converges define the \emph{Feynman period} of $G$ to be the projective integral
\begin{gather}\label{Pdef}
P_G = \int_{\alpha_e\geq 0} \frac{\Omega}{\Psi_G^2},
\end{gather}
where $\Omega = \sum_{i=1}^{|E(G)|} (-1)^i {\rm d}\alpha_1\cdots \widehat{{\rm d}\alpha_i}\cdots {\rm d}\alpha_{|E(G)|}$.
The Feynman period $P_G$ is the residue of its Feynman integral. It contributes to the $\beta$-function which controls the way the physical coupling changes with momentum, see, e.g.,~\cite{IZ}.
The $\beta$-function also plays a prominent role for applications to phase transitions in statistical physics, see, e.g.,~\cite{KP, ZJ}.
So, the Feynman period is of significance in various branches of physics. Furthermore, the Feynman period exhibits the arithmetic content of the Feynman integral, so it is very interesting for anyone
studying the geometry and number theory underlying quantum field theory~\cite{BEK,BK,PScoaction,Scensus}.

In~\cite{Brbig} Francis Brown introduced denominator reduction as a tool for studying Feynman periods algebraically. Denominator reduction describes how the denominators of the Feynman periods change through successive edge integrations. When these denominators fail to factor into linear pieces denominator reduction ends. Keeping track of the numerators in this process using multiple polylogarithms gives an algorithm for parametric Feynman integration which fails when denominator reduction stops~\cite{Brbig, BrH1}.
Improvements and variants of this integration approach have been implemented~\cite{Panzer:HyperInt}, but the constraint of the non-factoring denominators remains.

Denominator reduction itself is purely algebraic, but still carries substantial information about these integrals since the numerator polynomials in each integration step are typically of simpler
geometry than the denominator. This simplicity of the numerators is inherited from the second integration step, where the polynomials in the numerator are mere coefficients of $\Psi_G$
whereas the denominator is a resultant. The geometric domination of the denominator, however, is not a fully general feature. There exist (rare) examples where the Feynman integral $P_G$
has a geometry which is missed by denominator reduction.

Denominator reduction tells us about weight drop~\cite{BrY} and can be used to compute the $c_2$ invariant, an arithmetic graph invariant that predicts properties of
the Feynman periods~\cite{K3, SFq} (for a definition see Section~\ref{sec c2}).

In~\cite{Sc2} one of us defined a generalized denominator reduction, called \emph{quadratic denominator reduction}, that can always progress at least one more step than standard denominator reduction, and sometimes much further. At the cost of not working for even prime powers $\neq2$, this quadratic denominator reduction can be used to compute $c_2$ invariants and so still tells us about the geometries underlying the Feynman periods.

The geometric idea behind quadratic denominator reduction is that in certain cases a denominator with a square root still has rational geometry as it defines a projective line.

\begin{figure}[t]\centering
 \includegraphics{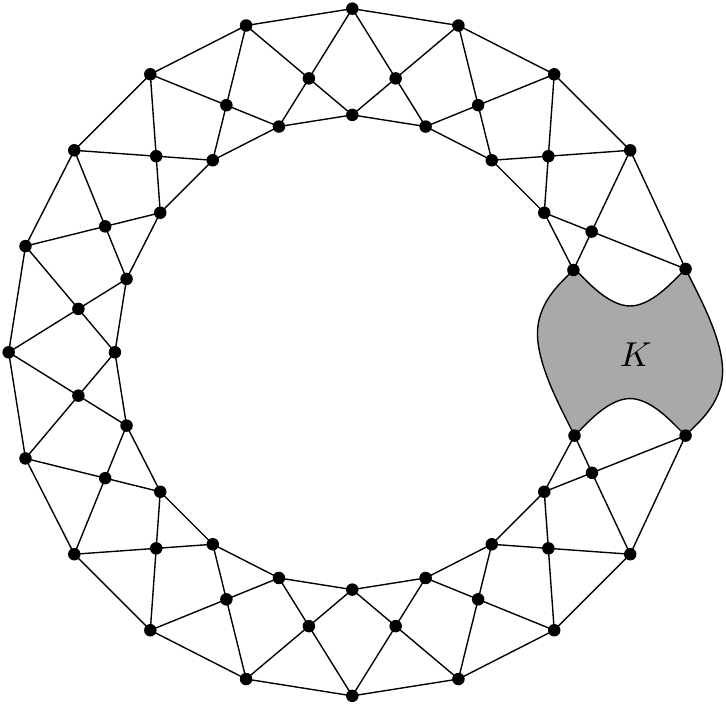}
 \caption{An illustration of the type of graph we consider. The shaded section is the kernel which is attached to the hourglass chain.}\label{fig graph type}
\end{figure}

As a demonstration of the power of quadratic denominator reduction, in this paper we will study infinite families of graphs built by attaching a chain of hourglasses to a finite kernel,
see~Figure~\ref{fig graph type}. The $c_2$ invariant for any such graph will only depend on the kernel, giving arbitrarily many infinite families of graphs for which we know the $c_2$ invariant
(for the prime 2 and odd prime powers).

With the notions of hourglass, kernel, and hourglass chain as illustrated in Figures~\ref{fig graph type} and \ref{fig Kprime}
we can state our theorem.
\begin{Theorem}\label{mainthm}
 Let $K$ be a kernel and let $L\in \mathcal{G}_K$ be a graph of the type we consider with an hourglass chain of length at least $6$, see Section~$\ref{sec hourglasses}$ for definitions.
 Let $v$ be a vertex of the hourglass chain that is shared by the second and third hourglasses from one end.
 Let $K'$ be $K$ with new edges joining the external vertices as in Figure~$\ref{fig Kprime}$.
 Index these new edges by $1$ and $2$ (with variables $\alpha_1$ and $\alpha_2$, respectively).
 If $q$ is an odd prime power then
 \begin{gather}\label{maineq}
 c_2^{(q)}(L-v)\equiv\big(\alpha_1\big(\Psi_{K'}^{1,2}(\alpha)\big)^2\Psi_{K'}^{2,2}(\alpha)\Psi_{K', 2}(\alpha)\big)_q\mod q,
 \end{gather}
 where $(F)_q$ is the Legendre sum of the polynomial $F$, see Definition~$\ref{deflege}$.
 For $q=2$ the $c_2$ invariant of $L-v$ vanishes, $c_2^{(2)}(L-v)\equiv0\mod 2$.
\end{Theorem}
The {\em Dodgson} polynomials on the right hand side of (\ref{maineq}) are defined in Section~\ref{sec background}.
Good introductory references for the general role of Dodgson polynomial in physics are, e.g.,~\cite{Brbig, Denham, Patt}.

\begin{figure}
\centering
 \includegraphics{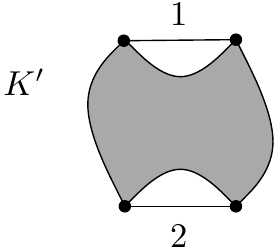}
 \caption{The graph $K'$ is the kernel $K$ with the two extra edges 1 and 2.}\label{fig Kprime}
\end{figure}

\begin{figure}[t]
\centering
\includegraphics[scale=.7]{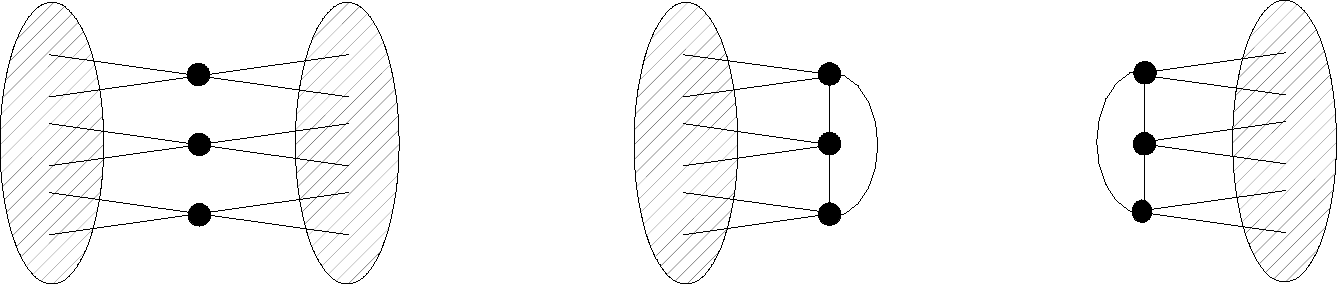}
\put(-240,24){\makebox(0,0)[lb]{$\phantom{P(\bar{G})}\qquad\quad\ \ =\qquad\quad \phantom{P(\bar{G}_1)}\qquad\times$}}
\put(-240,-20){\makebox(0,0)[lb]{$P(\bar{G})\qquad\quad\ \ =\qquad\quad P(\bar{G}_1)\qquad\times\qquad P(\bar{G}_2)$}}
\caption{Vertex connectivity 3 leads to a product of periods.}
\label{fig product}
\end{figure}

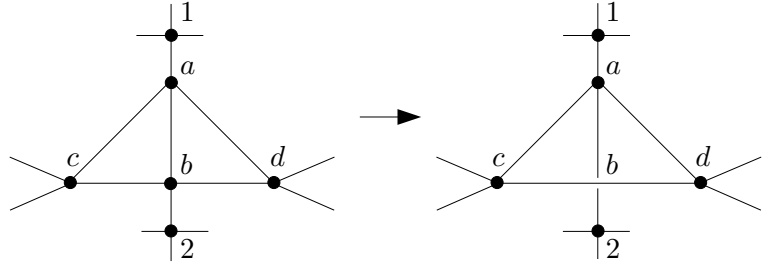
\begin{figure}[t]
\centering
\includegraphics[scale=0.7]{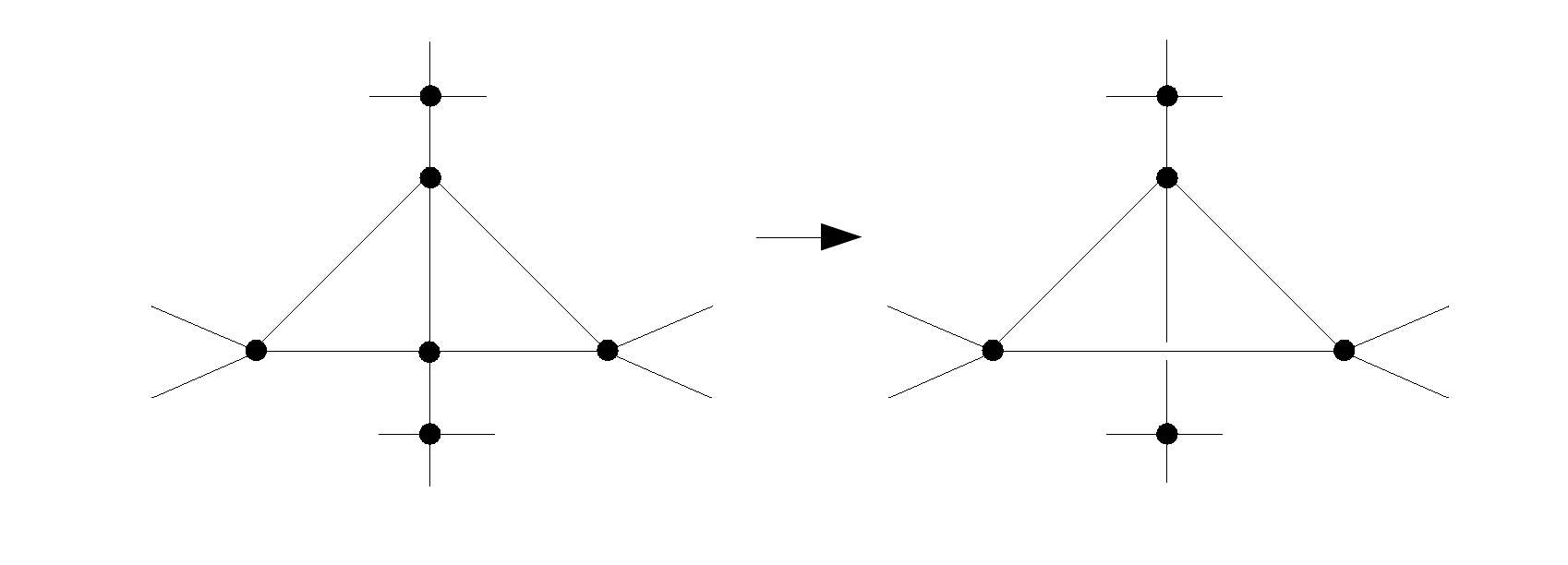}
\put(-222,92){\makebox(0,0)[lb]{$1$}}
\put(-222,72){\makebox(0,0)[lb]{$a$}}
\put(-222,34){\makebox(0,0)[lb]{$b$}}
\put(-265,37){\makebox(0,0)[lb]{$c$}}
\put(-188,37){\makebox(0,0)[lb]{$d$}}
\put(-61,92){\makebox(0,0)[lb]{$1$}}
\put(-61,72){\makebox(0,0)[lb]{$a$}}
\put(-61,34){\makebox(0,0)[lb]{$b$}}
\put(-104,37){\makebox(0,0)[lb]{$c$}}
\put(-27,37){\makebox(0,0)[lb]{$d$}}
\put(-222,2){\makebox(0,0)[lb]{$2$}}
\put(-61,2){\makebox(0,0)[lb]{$2$}}
\caption{Double triangle reduction: Replace a joint vertex of two attached triangles by a crossing.}
\label{fig double triangle reduction}
\end{figure}

Graphs $G$ for which the Feynman period exists, sometimes have structures which lead to graphical reductions, see Section~\ref{sec phi4} and~\cite{Scensus}.
If $G$ has a three vertex split the period factorizes, see Figure \ref{fig product}.
Double triangles can be reduced without changing the $c_2$~\cite{BSYc2}, see Figure \ref{fig double triangle reduction}.
Hourglass chains (for suitable kernels) have no such reductions. They establish families of the most complicated type, the {\em prime ancestors}~\cite{gfe, Scensus}.
Hourglass chains are the first families of prime ancestors for which
the $c_2$ invariant can be calculated. In a certain sense these hourglass chains can be considered as `telescopes' that enable us to look into geometries of
Feynman graphs at very high loop order (i.e., the number of independent cycles in $L-v$).
This has never been achieved before: all previous techniques were either restricted to the analysis of small graphs, or they worked in a way which was fundamentally prime-by-prime~\cite{CYgrid, Ycirc, Ystudy} and hence did not lead to non-trivial graph families with the same underlying geometries.

The paper is organized as follows: in Section~\ref{sec background} we provide the necessary background information on denominator reduction and the $c_2$ invariant.
Section~\ref{sec reductions} contains the proof of Theorem~\ref{mainthm}: the hourglass reductions.

Finally we use Theorem~\ref{mainthm} in Section~\ref{sec kernels} for an exhaustive search for $c_2$ invariants in hourglass chains with kernels of at most six internal vertices
(vertices which are not attached to the hourglasses, see Figure \ref{Fig:cases} and Table \ref{Tab:c2} for a maximum of five internal vertices).
We find Legendre symbols (see Section~\ref{sec background}) $(4/q)$ and $(-4/q)$ along with several modular forms. Explicitly the weight and level of the identified modular forms
(given in the notation [weight,level], see~\cite{BSmod,Sc2}) are [4,8], [4,16], [6,4], and [9,4]. The modular form [9,4] is a new addition
to the table of modular forms in $\phi^4$ theory as it was not found in $\phi^4$ graphs of loop orders less or equal twelve which were studied in~\cite{Sc2}.

An important outcome of this article is providing further support for the conjecture that in $\phi^4$ theory (corresponding to 4-regular graphs) the $c_2$ is free of curves
(which correspond to weight 2 modular forms), \cite[Conjecture~26]{BSmod}, see also~\cite{Sc2}. This puzzling conjecture is for the first time tested to any
loop order for some non-trivial geometries. It seems to be connected to some deep algebraic structure in quantum field theories. Note that curves in $c_2$ invariants are ubiquitous
if one lifts the (physical) restriction to 4-regular graphs.

For extra support of this ``no-curves-puzzle'' it might be worthwhile to study kernels in the future which lead out of $\phi^4$ graphs. Which non-$\phi^4$ kernels provide $c_2$ invariants
that correspond to weight two modular forms? It is also possible to extend the $c_2$-search to the rapidly increasing number of kernels with more than six internal vertices, see
Table~\ref{tab:numberkernels} and Question~\ref{quest}.

\section{Background}\label{sec background}

\subsection{Dodgsons}

In order to define the denominator reductions we first need to give a determinantal expression for $\Psi_G$ and define some related polynomials.

Assume for the rest of the paper that $G$ is a connected graph.

Choose an order on the edges $E(G)$ and vertices $V(G)$ of $G$ and choose a direction for each edge. Then the \emph{signed incidence matrix} of $G$ is a $|V(G)|\times|E(G)|$ matrix with entries $-1$, $0$, $1$, where the $i,j$th entry is $-1$ if edge $j$ starts at vertex $i$, is $1$ if edge $j$ ends at vertex $i$, and is $0$ otherwise. Let $E$ be the signed incidence matrix with one row removed. Since $G$ is connected the rank of the signed incidence matrix is $|V(G)|-1$ and so $E$ is full rank.

Define the \emph{expanded Laplacian} to be the matrix
\begin{gather*}
L_G =
\begin{bmatrix}
 \Lambda & E^t \\ E & 0
\end{bmatrix}\!,
\end{gather*}
where $\Lambda$ is the diagonal matrix with diagonal entries $\alpha_e$, $e\in E(G)$, in the edge order chosen above.
 This matrix is called the expanded Laplacian because it behaves very much like
the Laplacian (with a matching row and column removed), but the pieces of it have been expanded out into a larger block matrix.

\begin{Proposition}\label{prop psi mt}
 \begin{gather*}
 \Psi_G = (-1)^{|V(G)|-1}\det L_G.
 \end{gather*}
\end{Proposition}

This proposition is at its core the matrix-tree theorem~\cite[Section~2.2]{Brbig}. More specifically, the form of the matrix tree theorem that is most useful here is the form that says given a~\mbox{$|V(G)-1|\times |V(G)-1|$} submatrix of $E$, this matrix has determinant $\pm 1$ if the edges corresponding to the columns of the submatrix form a spanning tree of $G$ and has determinant $0$ otherwise.

In the following we also need the extension to minors of $L_G$. Polynomials from minors of $L_G$ are called {\em Dodgson} polynomials in~\cite{Brbig}. In general, these polynomials
have signs which depend on the sequence of edges~\cite{Sc2}. While in classical denominator reduction signs are often insignificant, they play an important role in quadratic denominator
reduction. By the special structure of~(\ref{maineq}) it is sufficient for our purpose to define the sign of Dodgson polynomials in a~trivial case.\looseness=1

\begin{Definition}%\label{def dodgson}
For any subsets $I$, $J$, $K$ of the edges of a connected graph $G$ with $|I|=|J|$ we define the Dodgson polynomials $\Psi^{I,J}_K$ as
 \begin{gather*}
 \Psi^{I,J}_{G,K} = \pm\det L_G^{I,J}\big|_{\alpha_k=0,k\in K},
 \end{gather*}
where $L_G^{I,J}$ is $L_G$ with rows in $I$ and columns in $J$ deleted. In the case $I=J$ the sign is $(-1)^{|V(G)|-1}$.
\end{Definition}

The contraction-deletion formula~(see~\cite{Brbig}, \cite[Lemma~11]{Sc2})
\begin{gather}\label{cd}
\Psi^{I,J}_{G,K}=\alpha_e\Psi^{Ie,Je}_{G,K}+\Psi^{I,J}_{G,Ke}=\alpha_e\Psi^{I,J}_{G\backslash e,K}+\Psi^{I,J}_{G/e,K}
\end{gather}
relates Dodgson polynomials to minors. Note that in the context of Dodgson polynomials graphs may have multiple edges and self-loops (which contract to zero, i.e., every
Dodson polynomial of a graph with a contracted self loop vanishes).

\begin{Example}
A tree has the graph polynomial 1. The Dodgson polynomial of a circle $C$ is the sum of its edge-variables. In this case $\Psi_C^{e,f}=\pm1$ for any two edges $e,f\in E(G)$.
\end{Example}

We get the following vanishing cases (see~\cite{Brbig}, \cite[Section~2.2 (4)]{BSYc2},
\cite[Lemma 13]{Sc2}):
\begin{gather}
\Psi^{I,J}_{G,K}=0\quad \text{if $I$ or $J$ cut $G$,}\nonumber
\\
\Psi^{I,J}_{G,K}=0\quad \text{if $(I\cup K)\backslash J$ or $(J\cup K)\backslash I$ contain a cycle.}
\label{vanishingcases}
\end{gather}

An important Dodgson identity is (see~\cite{Brbig} and~\cite[Lemma 18]{Sc2})
\begin{gather}\label{d12}
\Psi^{1,1}_{G,2}\Psi^{2,2}_{G,1}-\Psi^{12,12}_G\Psi_{G,12}=\big(\Psi^{1,2}_G\big)^2
\end{gather}
for any two edges 1, 2 in $G$. Many more identities for Dodgson polynomials can be found in~\cite{Brbig,Sc2}.

\subsection{Spanning forest polynomials}

Dodgson polynomials can also be written in terms of spanning forests. To that end, given a~partition $P$ of a subset of the vertices of $G$, define the \emph{spanning forest polynomial}
\begin{gather*}
\Phi_G^P = \sum_{F}\prod_{e\not\in F}\alpha_e,
\end{gather*}
where the sum is over spanning forests with the property that there is a bijection between the parts of $P$ and the trees of the forest such that every vertex in a part of the partition
is in the corresponding tree.

For example, given the graph illustrated in Figure~\ref{fig sp for eg}, with the partition $P$ indicated by the shape of the large vertices, the corresponding spanning forest polynomial is $\alpha_3(\alpha_4\alpha_2 + \alpha_1\alpha_2+\alpha_1\alpha_5+\alpha_2\alpha_5)$. This technique of marking the partition by the shape of the large vertices will be used without further comment in the main hourglass chain reduction argument.
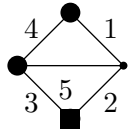
\begin{figure}[h]
\centering
 \includegraphics[scale=1.25]{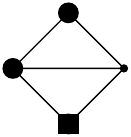}
 \put(-40,35){\makebox(0,0)[lb]{$4$}}
 \put(-10,35){\makebox(0,0)[lb]{$1$}}
 \put(-40,7){\makebox(0,0)[lb]{$3$}}
 \put(-10,7){\makebox(0,0)[lb]{$2$}}
 \put(-27,12){\makebox(0,0)[lb]{$5$}}
 \caption{An example of a graph with a partition of some of the vertices marked by the shape of the large vertices.}\label{fig sp for eg}
\end{figure}

We have the following proposition
\begin{Proposition}\label{prop d to s}
\begin{gather*}
\Psi_{G,K}^{I,J} = \sum f_k \Phi_{G\backslash (I\cup J\cup K)}^{P_k},
\end{gather*}
where the $P_k$ run over partitions of the ends of $I$, $J$, and $K$, and $f_k\in \{-1, 0, 1\}$.
Furthermore, those $f_k$ which are nonzero are exactly those where each forest in the polynomial becomes a tree in $G\backslash I /(J\cup K)$ and in $G\backslash J /(I\cup K)$; if any forest in the polynomial has this property then all of them do.
\end{Proposition}
This proposition is a particular interpretation of the all minors matrix tree theorem, for a~proof see~\cite[Propositions 8 and 12]{BrY}. To see the connection to the matrix tree theorem briefly, consider a term in the determinant giving $\Psi_{G,K}^{I,J}$. The variables indicate edges of $G$, where the corresponding columns are removed in~$E$ and corresponding rows are removed in~$E^t$. The rows indexed by $I$ are removed in~$E^t$ but not in~$E$ and the columns indexed by~$J$ are removed in~$E$ but not in~$E^t$. By the matrix tree theorem as summarized after Proposition~\ref{prop psi mt}, both these sets of edges must simultaneously be spanning trees. Furthermore the variables indexed by~$K$ are set to~$0$. So they must not be in the monomial and hence must be in the trees.
An edge which must be in a tree can be contracted, and splitting apart vertices which were identified via a contraction splits the tree into a forest with constraints on which vertices belong in which tree of the forest. Working out the details gives the result, see~\cite{BrY} for details.

The signs in the spanning forest polynomial expansion of a Dodgson polynomial can be tricky, however, the only case we will need for the argument below is given in the following lemma.
\begin{Lemma}\label{lem easy sign}
 With notation as in Proposition~$\ref{prop d to s}$, if $P_1$ and $P_2$ both have nonzero coefficients, and $P_1$ and $P_2$ differ by swapping two vertices which are in the same component of $J$ viewed as a subgraph of $G$, then $f_1 = -f_2$.
\end{Lemma}
This is a special case of~\cite[Corollary~17]{BrY}.

Another useful observation, see~\cite[Proposition 21]{BrY}, is that if $G$ is formed as the 2-sum of~$G_1$ and $G_2$, that is $G_1$ and $G_2$ each have a distinguished edge, $e_1$ and $e_2$ respectively, and $G$ is the result of identifying $e_1$ and $e_2$ and then removing this new identified edge while leaving the induced identifications on the incident vertices, then
\begin{gather}\label{eq 2 cut}
 \Psi_G = \Psi_{G_1\backslash e_1}\Psi_{G_2/e_2} + \Psi_{G_1/e_1}\Psi_{G_2\backslash e_2} = \Psi_{G_1\backslash e_1}\Phi_{G_2\backslash e_2}^{\{v_1\}, \{v_2\}} + \Phi_{G_1\backslash e_1}^{\{v_1\}, \{v_2\}}\Psi_{G_2\backslash e_2},
\end{gather}
where $v_1$ and $v_2$ are the ends of $e_1$ and $e_2$.

Note that all Dodgson polynomials and all spanning forest polynomials are explicitly linear in all their variables.

\subsection[The c2 invariant]
{The $\boldsymbol{c_2}$ invariant}\label{sec c2}

A general aim in the mathematical theory of Feynman periods is to understand what kind of numbers can appear~\cite{BK,Scensus,Snumfunct}. The interest in this topic was recently intensified by
the (conjectural) discovery of a Galois coaction structure on these numbers~\cite{Bcoact1,Bcoact2, PScoaction}. In general, the Feynman period (\ref{Pdef}) is hard to analyze.
Even the zero locus of the graph hypersurface $\Psi_G=0$ has a complicated geometric structure.

The number theoretic content of the Feynman period is intimately related to the motivic structure of its integral. In general, the motivic setup is a deep superstructure to the cohomology theory of
integrals in algebraic geometry. Going back to ideas of A. Grothendieck it lifts Galois theory to higher dimensions (for first reading we recommend~\cite{intmot}).
This motivic structure unifies all fields and thus finite fields $\FF_q$ encapture information on the geometry of the
Feynman period. The motivic information can be extracted from the number of elements on the singular locus of the integrand.
This is, e.g., visible in the action of the Frobenius homomorphism in the classical Lefschetz fixed point theorem~\cite{Lef}.
While the knowledge of this point-count for a single $q$ is still not very informative, its value for all (or many) $q$ carries important number theoretical information on the
({\em framing} of the) Feynman period~\cite{BD}.

We define
\begin{gather*}
[F]_q=|\{F=0\text{ in }\FF_q\}|
\end{gather*}
as the point-count of the zero locus of the polynomial $F$. In the context of Feynman periods the important information of the point-count is hidden in the first non-trivial reduction modulo $q$.
For any connected graph $G$ with at least three edges we define~\cite{SFq}
\begin{gather}\label{c2def}
c_2^{(q)}(G)\equiv\frac{[\Psi_G]_q}{q^2}\mod q
\end{gather}
as the $c_2$ invariant of the Feynman graph $G$. The above definition implies that the point-count of the graph hypersurface is always divisible by $q^2$ (the index 2 in $c_2$ refers to this square).
For a given graph $G$ one should think of the $c_2$ as the infinite sequence $\big(c_2^{(q)}\big)_{q=2,3,4,5,7,8,9,11,\dots}$ of remainders modulo $q$.

The benefit of the reduction modulo $q$ is that the point-count is combinatorially quite accessible. In practice, non-trivial prime powers are still harder to come by,
so that the $c_2$ invariant is often studied for pure primes only. It is conjectured \cite[Conjecture~2]{Sc2}, that the knowledge of the $c_2$ for all primes determines
the $c_2$ for all prime powers.

The $c_2$ has been studied quite deeply in the context of $\phi^4$ quantum field theory (Section~\ref{sec phi4}). The focus of these studies can either be
the general mathematical structure of the $c_2$~\cite{BSYc2,HSSYc2,Yscompl} or the zoology of the geometries identified by $c_2$s~\cite{K3,BSmod,CYgrid,Sc2,Ycirc,Ystudy}.
The nature of this article is more in the latter direction, particularly when we analyze the $c_2$s of small kernels in Section~\ref{sec kernels}. We would like to emphasize
that although identifying $c_2$s can have an experimental flavour it might be of high importance to understanding the algebraic structure of quantum field theories.

We can get rid of the division by $q^2$ in (\ref{c2def}) by using Dodgsons instead of the graph polynomial (see~\cite[Corollary~28 and Theorem~29]{K3})
\begin{gather*}%\label{c2def2}
c_2^{(q)}(G)\equiv-\big[\Psi^{13,23}_G\Psi^{1,2}_{G,3}\big]_q\mod q
\end{gather*}
for every connected graph $G$ with a degree 3 vertex $v$.
In this version the $c_2$ can be further simplified by denominator reductions and quadratic denominator reductions as will be outlined in the next sections.

\subsection{Denominator reduction}\label{sec dr}

Successive integration of the Feynman period (\ref{Pdef}) leads to denominators which are linear in the next integration variable. After three initial steps the new denominator
will be the resultant of the old denominator with respect to the integration variable. Eventually the denominator may cease to factor into linear pieces and then one typically enters
very complicated territory. This successive taking of resultants has a point-count version which says that under certain conditions
\begin{gather*}
[(A\alpha+B)(C\alpha+D)]_q\equiv-[AD-BC]_q\mod q.
\end{gather*}
If the resultant $AD-BC$ factors in some new variable then the reduction can be repeated. The size of the polynomial that has to be counted reduces rapidly and when no more reduction
is possible then one can still resort to brute force counting at the last step. More precisely, we obtain the following result.

\begin{Definition}[{denominator reduction \cite[Definition~120 and~Proposition 126]{Brbig}}]\label{defdr}
Given a~connected graph $G$ with at least three edges and a sequence of edges $1,2,\dots,{|E(G)|}$ we define
\begin{gather*}
^3\Psi_G(1,2,3)=\pm\Psi^{13,23}_G\Psi^{1,2}_{G,3}.
\end{gather*}
Suppose $^n\Psi_G$ for $n\geq3$ factors as
\begin{gather*}
^n\Psi_G(1,\dots,n)=(A\alpha_{n+1}+B)(C\alpha_{n+1}+D)
\end{gather*}
then we define
\begin{gather*}
^{n+1}\Psi_G(1,\dots,n+1)=\pm(AD-BC).
\end{gather*}
Otherwise denominator reduction terminates at step $n$. If it exists we call $^n\Psi_G$ an $n$-invariant of $G$.
\end{Definition}
Note that the $n$-invariants are only defined up to sign. The 4-invariant always factorizes \cite[Lemma~82]{Brbig}
\begin{gather}\label{4inv}
^4\Psi_G=\pm\Psi^{14,23}_G\Psi^{13,24}_G.
\end{gather}
Therefore the 5-invariant always exists. In Lemma 87 of~\cite{Brbig} it is proved that for $n\geq5$ the $n$-invariants become independent of the sequence of the reduced edges
(they only depend on the set of reduced variables). Denominator reduction is compatible with the $c_2$ invariant in the following sense.

\begin{Theorem}[{\cite[Theorem~29]{K3}}]\label{thmdr}
Let $G$ be a connected graph with at least three edges and $h_1(G)\leq|E(G)|/2$ independent cycles. Then
\begin{gather*}%\label{c2frominvariant}
c_2^{(q)}(G)\equiv (-1)^n [^n\Psi_G]_q\mod q
\end{gather*}
whenever $^n\Psi_G$ exists for $n<|E(G)|$.
\end{Theorem}

The theorem was proved for $\geq5$ edges in~\cite{K3}. The proof trivially extends to the case of three or four edges.
If $^n\Psi_G=0$ for some sequence of edges and some $n$ (and hence for all subsequent~$n$) then $G$ has {\em weight drop}. In this case the $c_2$ invariant vanishes~\cite{BrY}.

\subsection{Quadratic denominator reduction}\label{sec qdr}

Only in a few particularly simple cases does denominator reduction go through to the very end where all variables are reduced. Brute force point-counting after the last step of denominator
reduction can be very time consuming and does not lend itself to more theoretical understan\-ding. Therefore it is desirable to continue the reduction as far as possible. The integration of a~denominator that does not factor produces a square root. The existence of further reduction steps is suggested by the fact that even in the presence of squares and square roots integrals may stay
rational in a geometrical sense. Let us exemplify this by the following toy integrals~\cite{Sc2},
\begin{gather*}
\int_0^\infty\frac{\dd\alpha}{A\alpha^2+B\alpha+C}=\frac{\log(X)}{\sqrt{B^2-4AC}},
\end{gather*}
for some algebraic expression $X$ in $A$, $B$, $C$ and
\begin{gather*}
\int_0^\infty\frac{\dd\alpha}{\sqrt{D\alpha^2+E\alpha+F}(H\alpha+J)}=\frac{\log(Y)}{\sqrt{DJ^2-EHJ+FH^2}}
\end{gather*}
for some algebraic expression $Y$ in $D$, $E$, $F$, $H$, $J$.

For the formal implementation of this idea we pass from point-counts to Legendre sums.

\begin{Definition}[{\cite[Definition 29]{Sc2}}]\label{deflege}
Let $q$ be an odd prime power. For any $a\in\FF_q$ the Legendre symbol $(a/q)\in\{-1,0,1\}$ is defined by
\begin{gather*}%\label{legedef}
\bigg(\frac{a}{q}\bigg)=\big|\big\{x\in\FF_q\colon x^2=a\big\}\big|-1.
\end{gather*}
For any polynomial $F\in\ZZ[\alpha_1,\dots,\alpha_N]$ we define
\begin{gather*}%\label{Fqdef}
(F)_q=\sum_{\alpha\in\FF_q^N}\bigg(\frac{F(\alpha)}{q}\bigg),
\end{gather*}
where the sum is in $\ZZ$.
\end{Definition}
The Legendre symbol is multiplicative, $(ab/q)=(a/q)(b/q)$ for $a,b\in\FF_q$ and trivial for squares $\big(a^2/q\big)=1-\delta_{a,0}$, where the delta is the characteristic function on $\FF_q$.
This leads to $\big(F^2\big)_q=q^N-[F]_q$ for any polynomial $F$ in $N$ variables. If $N\geq1$ we get
\begin{gather*}%\label{point2Legendre}
[F]_q\equiv-\big(F^2\big)_q\mod q.
\end{gather*}
With the above equation we can translate point-counts to Legendre sums. Quadratic denominator reduction knows two cases,
\begin{gather*}
\big(\big(A\alpha^2+B\alpha+C\big)^2\big)_q\equiv-\big(B^2-4AC\big)_q\mod q,
\\
\big(\big(D\alpha^2+E\alpha+F\big)(H\alpha+J)^2\big)_q\equiv-\big(DJ^2-EHJ+FH^2\big)_q\mod q
\end{gather*}
if the total degree of the polynomials on the left hand sides does not exceed twice the number of their variables. The proof of these identities is in~\cite[Section~7]{Sc2}.
It uses a Chevalley--Warning-Ax theorem for double covers of affine space which is proved by F.~Knop in~\cite[Appendix]{Sc2}. As~examples, $(W\alpha+X)^3(Y\alpha+Z)$ reduces by
the second case to zero whereas $\big(U\alpha^2+V\alpha\allowbreak+W\big)\big(X\alpha^2+Y\alpha+Z\big)$ or $(U\alpha+V)(W\alpha+X)(Y\alpha+Z)$ do not reduce in general.

Note that in the case that both quadratic reductions are applicable one is back to the case of standard denominator reduction. Then both reductions lead to the same result.

We define quadratic $n$-invariants $^n\Psi^2_G$ in analogy to Definition \ref{defdr}.

\begin{Definition}[{quadratic denominator reduction~\cite[Definition 34]{Sc2}}]%\label{defqdr}
{\sloppy
Given a connected graph~$G$ with at least three edges and a sequence of edges $1,2,\dots,{|E(G)|}$ we define
\begin{gather}\label{n3}
^3\Psi^2_G(1,2,3)=\big(^3\Psi_G(1,2,3)\big)^2.
\end{gather}}\noindent
Suppose $^n\Psi^2_G$ for $n\geq3$ is of the form
\begin{gather}\label{case1}
^n\Psi^2_G(1,\dots,n)=\big(A\alpha_{n+1}^2+B\alpha_{n+1}+C\big)^2
\end{gather}
then we define
\begin{gather*}
^{n+1}\Psi^2_G(1,\dots,n+1)=B^2-4AC.
\end{gather*}
Suppose $^n\Psi^2_G$ is of the form
\begin{gather}\label{case2}
^n\Psi^2_G(1,\dots,n)=\big(D\alpha_{n+1}^2+E\alpha_{n+1}+F\big)(H\alpha_{n+1}+J)^2
\end{gather}
then we define
\begin{gather*}
^{n+1}\Psi^2_G(1,\dots,n+1)=DJ^2-EHJ+FH^2.
\end{gather*}
Otherwise quadratic denominator reduction terminates at step $n$. If it exists we call $^n\Psi^2_G$ a~quadratic $n$-invariant of $G$.
If $^n\Psi^2_G=0$ for some sequence of edges and some $n$ then we say that~$G$ has weight drop.
\end{Definition}

Note that quadratic $n$-invariants have no sign ambiguity. The connection to the $c_2$ invariant is similar to the standard case, with a restriction to $q=2$ or odd prime powers $q$.
\begin{Theorem}[{\cite[Theorem~36 and Remark 37]{Sc2}}]\label{thmqdr}
Let $q$ be an odd prime power and $G$ be a connected graph with at least three edges and $h_1(G)\leq|E(G)|/2$ independent cycles. Then
\begin{gather*}%\label{c2fromquadratic}
c_2^{(q)}(G)\equiv (-1)^{n-1}\big({}^n\Psi^2_G\big)_q\mod q
\end{gather*}
whenever $^n\Psi^2_G$ exists. If $^n\Psi^2_G\equiv0\mod2$ then $c_2^{(2)}(G)\equiv 0\mod 2$.
\end{Theorem}
If the $n$-invariant $^n\Psi_G$ exists, we get $^n\Psi^2_G=[^n\Psi_G]^2$, generalizing (\ref{n3}). In many cases quadratic denominator reduction goes significantly beyond
standard denominator reduction.

\subsection{Scaling}

Even when quadratic denominator reduction stops it is often possible to simplify further by scaling some variables (see Section~\ref{sec endgame}). In the context of
Legendre sums this technique is based on the elimination of square factors from the Legendre symbol. For classical point-counts scaling was already used in~\cite{SFq} and later in~\cite{K3}, with some observations on some combinatorial conditions which allow it in~\cite{Ysome}.
In this article we have a case where the result after scaling can be further reduced by additional steps of quadratic denominator reduction. This makes the scaling technique particularly powerful.

\subsection[phi4 theory]
{$\boldsymbol{\phi^4}$ theory}\label{sec phi4}

The most interesting graphs for us are the primitive 4-point $\phi^4$ graphs. Rephrased in a purely graph theoretic language, this means we are most interested in graphs which can be obtained by taking a 4-regular graph and removing one vertex. The 4-regular graph needs to be internally 6-edge connected, that is the only 4-edge cuts of the 4-regular graph are those which separate one vertex from the rest of the graph. For graphs obtained from an internally 6-edge connected 4-regular graph in this way the Feynman period is convergent~\cite{Scensus}. For such graphs we have
(quadratic) denominator reduction, Theorems \ref{thmdr} and~\ref{thmqdr} as well as the Chevalley--Warning theorem as extra tool (see~\cite[Lemma 2.6]{Ycirc} for an exposition).

Furthermore, the Feynman period of a graph obtained by removing a vertex of an internally 6-edge connected 4-regular graph does not depend on the choice of vertex removed.
This is the \emph{completion invariance} of the Feynman period. The analogous invariance for the $c_2$ invariant is conjectural~\cite{K3}, but an approach based on counting edge partitions has enabled
a proof when $q=2$ and the $4$-regular graph has an odd number of vertices~\cite{Yscompl}. Upcoming work of one of us with Simone Hu will complete the $q=2$ proof. In our hourglass chain graphs,
we will be removing the most convenient vertex; if we assume the conjecture then this is equivalent to removing any other vertex.

With the completion conjecture we can also ignore 4-regular graphs with three vertex splits ({\em reducible} graphs in~\cite{Scensus}, see Figure \ref{fig product}):
By deleting one of the three split vertices the decompleted graph inherits a two-vertex split which renders the $c_2$ trivial~\cite{BSYc2}. The Feynman period of a~graph with a 3-vertex split factorizes~\cite{Scensus}. Another reduction is obtained by ignoring graphs with double triangles (a pair of triangles with a common edge, see Figure \ref{fig double triangle reduction}).
It was shown in~\cite{BSYc2} that double triangles can be reduced
to single triangles (one of the common vertices becomes a crossing) without changing the $c_2$ invariant. With all reductions (internally 6-connected, 3-vertex connected, double-triangle-free)
we are lead to considering prime ancestors~\cite{gfe, Scensus}. Note that for suitable kernels $K$ the hourglass chains in this paper provide infinite families of prime ancestors.

\subsection{Hourglasses}\label{sec hourglasses}

By an \emph{hourglass} we mean two triangles sharing one common vertex. Taking two of the degree two vertices of an hourglass which are not in the same triangle and joining to two such vertices in another hourglass, we obtain a \emph{bihourglass}, see Figure~\ref{fig hourglass}.
Continuing by joining a third hourglass in the same way to the remaining degree 2 vertices of the second hourglass, and so on, we obtain longer \emph{hourglass chains}, where the hourglass and bihourglass are the hourglass chains of length~1 and~2 respectively.\vspace{2ex}
\begin{figure}[h!]
\centering
 \includegraphics[scale=1.1]{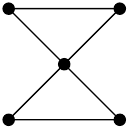} \quad \includegraphics[scale=1.05]{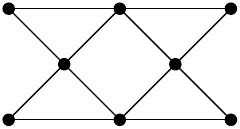}
 \caption{An hourglass and a bihourglass.}\label{fig hourglass}
\end{figure}

Hourglass chains of any length have four degree 2 vertices and all remaining vertices of degree~4. Take another fixed graph $K$ which has four degree 2 vertices and all remaining vertices of degree 4. Additionally, fix a bipartition of the degree 2 vertices of $K$ into two parts of size $2$. We will call $K$ (with this choice of bipartition) the \emph{kernel}, see Figure \ref{Fig:cases}.
Let $\mathcal{G}_K$ be the family of graphs obtained by taking an hourglass chain of any length, joining the two degree~2 vertices at one end of the chain to the two degree 2 vertices in one part
of the bipartition given with $K$, and joining the 2 vertices at the other end of the chain to the two degree~2 vertices in the other part of the bipartition. See Figure~\ref{fig graph type} for
an illustration. Note there are two ways to join on any hourglass chain compatible with the bipartition, differing by a half (M\"obius) twist.
In the end, this half twist will not affect the $c_2$ invariant, and so we include both in~$\mathcal{G}_K$.\looseness=1

\section{Hourglass reductions}\label{sec reductions}

The goal of this section is to prove Theorem~\ref{mainthm}. Note that after using the theorem one can continue to reduce any variable of the kernel $K$.
If $K$ is the kernel of a 4-regular hourglass chain then $K'$ has four vertices of degree three at the ends of the extra edges 1 and 2.
The topology of degree three vertices simplifies the structure of related Dodgsons~\cite{Brbig,K3}, \cite[Lemma~19]{Sc2}.
Particularly simple is the case of edge 2 whose variable is absent in (\ref{maineq}) (note that the choice of labels 1 and 2 is arbitrary).
Experiments suggest that it might always be possible to reduce both edges $\neq2$ of any vertex adjacent to edge 2.

\begin{question}\label{quest}
Let $v$ be a degree three vertex attached to edge 2 in $K'$. Let $\alpha_3$ and $\alpha_4$ be the edges of $v$ which are in $K$. Is it always possible to quadratically denominator
reduce the right hand side of (\ref{maineq}) with respect to $\alpha_3$ and $\alpha_4$? If yes, what expression does one get after the quadratic reduction of $\alpha_3$ and $\alpha_4$?
\end{question}

In general it is quite helpful for analyzing larger kernels to have closed expressions for the $c_2$ with as many reductions as possible.

Figure~\ref{fig overview} gives an overview of how the reductions will proceed.
Following the specified order in an explicit example using a computer for the reductions can also be a helpful way to follow through the general argument.

\begin{figure}
\centering
 \includegraphics[scale=.95]{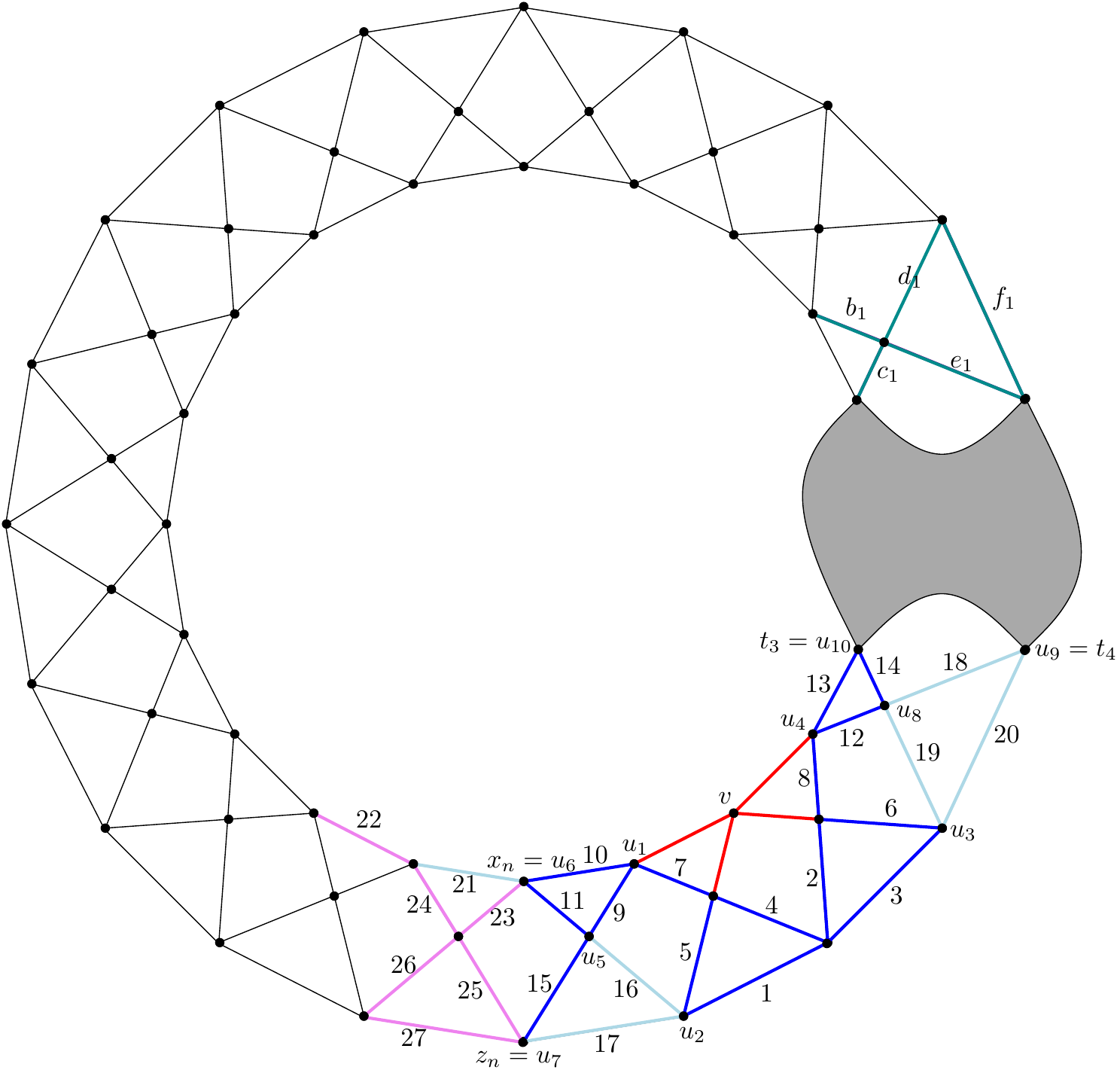}
 \caption{The order in which the edges should be reduced. First decomplete at $v$, thus removing the red edges. Then reduce the dark blue edges with conventional denominator reduction. Next reduce the light blue edges with quadratic denominator reduction. Continue with quadratic denominator reduction to reduce the six pink edges. The same form of denominator reappears, and so inductively we can reduce the analogous six edges in each subsequent hourglass until we reach the last hourglass. Finally reduce the remains of the last hourglass, the five dark cyan edges, according to Section~\ref{sec endgame}.}\label{fig overview}
\end{figure}

\subsection{Initial reductions}

The first step of the proof is to begin a conventional denominator reduction on $L-v$, see Section~\ref{sec dr}.

Consider the two hourglasses around $v$ and label the edges as in Figure~\ref{fig decompletion setup}. Then beginning our denominator reduction at the 4-invariant with (see (\ref{4inv}))\vspace{-.5ex}
\begin{gather*}
\Psi_{L-v}^{14,23}\Psi_{L-v}^{13,24},
\end{gather*}
the reductions of $5$ and $6$ are forced to avoid contracting the triangles $145$ and $236$ in the first factor, and then the reductions of $7$ and $8$ are forced
to avoid disconnecting the degree three vertices $457$ and $268$ in the first factor, see (\ref{vanishingcases}). This yields\vspace{-.5ex}
\begin{gather*}
\Psi^{1456, 2356}_{L-v, 78}\Psi^{1378,2478}_{L-v, 56}.
\end{gather*}

\begin{figure}
\centering
 \includegraphics{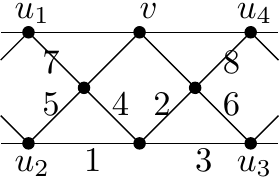}
 \caption{Edge labelling around $v$.}\label{fig decompletion setup}
\end{figure}

At this point it is more convenient to consider the situation in terms of spanning forest polynomials via Proposition~\ref{prop d to s} and Lemma~\ref{lem easy sign},\vspace{-.5ex}
\begin{gather*}
\Psi^{1456, 2356}_{L-v, 78}\Psi^{1378,2478}_{L-v, 56}
= \pm\Big(\Phi_{L'}^{\{u_1, u_4\}, \{u_2, u_3\}} - \Phi_{L'}^{\{u_1, u_3\}, \{u_2, u_4\}}\Big)\Phi_{L'}^{\{u_2\}, \{u_3\}},
\end{gather*}
where $L'$ is $L$ without $v$, the edges $1$ through $8$, and without the three vertices isolated by those removals. Label the triangle of $L'$ containing $u_1$ as in Figure~\ref{fig label 9 10 11}. Reduce edge $9$ by the general deletion and contraction reduction formula (\ref{cd}).
Notice that in both terms where $9$ was deleted, $10$ is an isthmus and $u_1$, the vertex at the isolated end of $10$, is not in a part by itself.
Thus $10$ cannot be cut in these terms, forcing $10$ to be contracted in the other factors. This gives\vspace{-.5ex}
\begin{gather*}
\pm\Big(\Phi_{L'-u_1}^{\{u_5, u_4\}, \{u_2, u_3\}} - \Phi_{L'-u_1}^{\{u_5, u_3\}, \{u_2, u_4\}} - \Phi_{L'-u_1}^{\{u_6, u_4\},\{u_2, u_3\}} + \Phi_{L'-u_1}^{\{u_6, u_3\}, \{u_2, u_4\}}\Big) \Phi_{L'-u_1}^{\{u_2\}, \{u_3\}},
\end{gather*}
where the vertices are as labelled in Figure~\ref{fig label 9 10 11}.

\begin{figure}
\centering
 \includegraphics{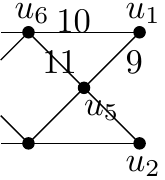}
 \caption{Edge labelling for the triangle containing $u_1$.}\label{fig label 9 10 11}
\end{figure}

Consider which trees the vertices $u_6$ and $u_7$ can belong to in the first factor of the previous expression. This factor is illustrated in Figure~\ref{fig reduce 10}. Both trees of the forest appear in the portion of the graph including vertices $u_2$, $u_5$, $u_6$ and $u_7$ and so must exit from this portion to the rest of the graph via the only possible vertices, $u_6$ and $u_7$. Since there are two trees and two vertices, one must use $u_6$ and the other must use $u_7$. In view of the shape of the graph, the only way this can happen is that the tree corresponding to the square vertices exits via~$u_7$ and the tree corresponding to the circle vertices exits via~$u_6$.
Then the only thing undetermined in how the trees go through the illustrated part of the graph, is the tree to which $u_5$ belongs in the third and fourth terms. Summing over both possibilities we see that one of the possibilities cancels with the first two terms, and so what remains of the entire expression is
\begin{gather*}
\pm\Big(\Phi_{L'-u_1}^{\{u_2, u_5, u_7, u_3\}\{u_6, u_4\}}-\Phi_{L'-u_1}^{\{u_2,u_5,u_7,u_4\}\{u_6,u_3\}}\Big)\Phi_{L'-u_1}^{\{u_2\},\{u_3\}}.
\end{gather*}

\begin{figure}\centering
 \includegraphics{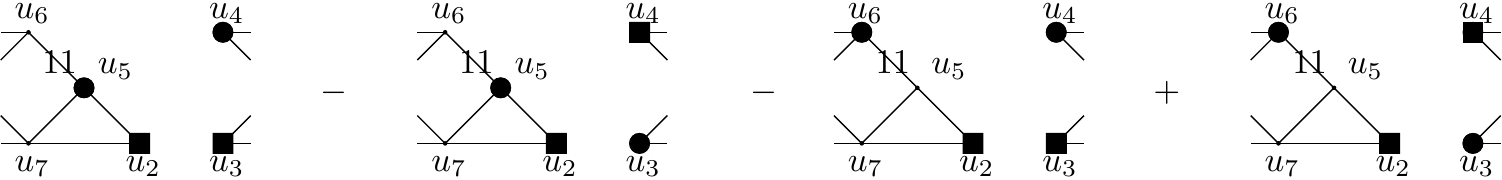}
 \caption{The first factor in the expression after reducing edge 10.}\label{fig reduce 10}
\end{figure}

In the first factor, edge 11 must always be deleted since its two ends are in different parts and so we obtain
\begin{gather*}
\pm\Big(\Phi_{(L'-u_1)\backslash 11}^{\{u_2, u_5, u_7, u_3\}\{u_6, u_4\}}-\Phi_{(L'-u_1)\backslash 11}^{\{u_2,u_5,u_7,u_4\}\{u_6,u_3\}}\Big)\Phi_{(L'-u_1)/11}^{\{u_2\},\{u_3\}}.
\end{gather*}
The only way that the hourglass with $u_3$ and $u_4$ affected this computation was to guarantee that both trees had to leave the part of the graph illustrated on the left. In other words, all we needed to know is that both trees appeared in the part of the graph illustrated on the right. Swapping left and right this remains true and so we can use the same argument as above on the triangle involving $u_4$. Labelling the vertices of the hourglass including $u_3$ and $u_4$ as in Figure~\ref{fig next hourglass}, this calculation gives
\begin{gather*}
\pm\Big(\Phi_{L_1}^{\{u_2, u_5, u_7, u_3, u_8, u_9\},\{u_6, u_{10}\}} - \Phi_{L_1}^{\{u_2, u_5, u_7, u_{10}\},\{u_3,u_8, u_9, u_6\}}\Big)\Phi_{L_2}^{\{u_2\},\{u_3\}},
\end{gather*}
where $L_1 = (L'-\{u_1,u_4\})\backslash \{11,14\}$ and $L_2 = (L'-\{u_1,u_4\})/ \{11,14\}$. Now comes the key observation that the first factor can be factored. First, the triangles $u_2$, $u_5$, $u_7$ and $u_3$, $u_8$, $u_9$ factor off since they are only joined at a vertex and the tree to which that vertex belongs is known. Additionally, similarly to the observations used above, both trees need to propagate through each hourglass remaining in the chain since both trees appear on both sides. However, the trees cannot cross within an hourglass, so one tree must run down one side of the hourglasses and the other tree down the other side; only the middle vertex could be in either tree.\looseness=1

\begin{figure}
\centering
 \includegraphics{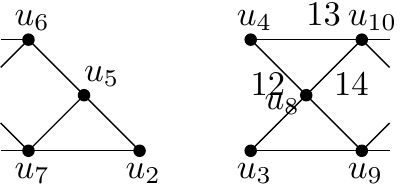}
 \caption{The vertex labellings for the hourglass containing $u_3$ and $u_4$.}\label{fig next hourglass}
\end{figure}

To write this down nicely we need some more systematic notation. The edges and vertices of each hourglass will be labelled as in Figure~\ref{fig hourglass labels}, where the hourglasses $X_i$ are indexed by $i$. When it is useful to talk about an hourglass generically we will leave out the subscripts. Additionally, let $t_1$, $t_2$, $t_3$ and $t_4$ be the degree 2 vertices of $K$ with the bipartition being $\{t_1, t_2\}$, $\{t_3, t_4\}$. With this notation, the factorization observation allows us to rewrite the denominator expression so far as
\begin{gather}
\pm\Psi_{L[u_2, u_5, u_7]}\Psi_{L[u_3,u_8,u_9]}\bigg(\prod_{i}\Phi_{X_i}^{\{w_i,x_i\}, \{y_i,z_i\}} \bigg)\Big(\Phi_{K}^{\{t_1, t_3\}, \{t_2, t_4\}} - \Phi_{K}^{\{t_1, t_4\}, \{t_2, t_3\}}\Big)\nonumber
\\ \qquad
{}\times\Phi_{L_2}^{\{u_2\}, \{u_3\}},
\label{eq the expression}
\end{gather}
where $L[S]$ for a set of vertices $S$ indicates the induced subgraph of $L$ given by the vertices of $S$, that is the subgraph with the vertices of $S$ and all edges in $L$ that have both ends in $S$. In this case the two induced subgraphs are both triangles.
Diagrammatically, this expression can be represented as in Figure~\ref{fig the expression}. The reader is encouraged to draw all the steps diagrammatically, as this gives the most insight into the calculation.

\begin{figure}
\centering
 \includegraphics{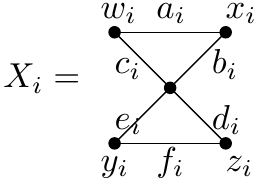}
 \caption{Labels for hourglasses.}\label{fig hourglass labels}
\end{figure}

\begin{figure}
\centering
 \includegraphics{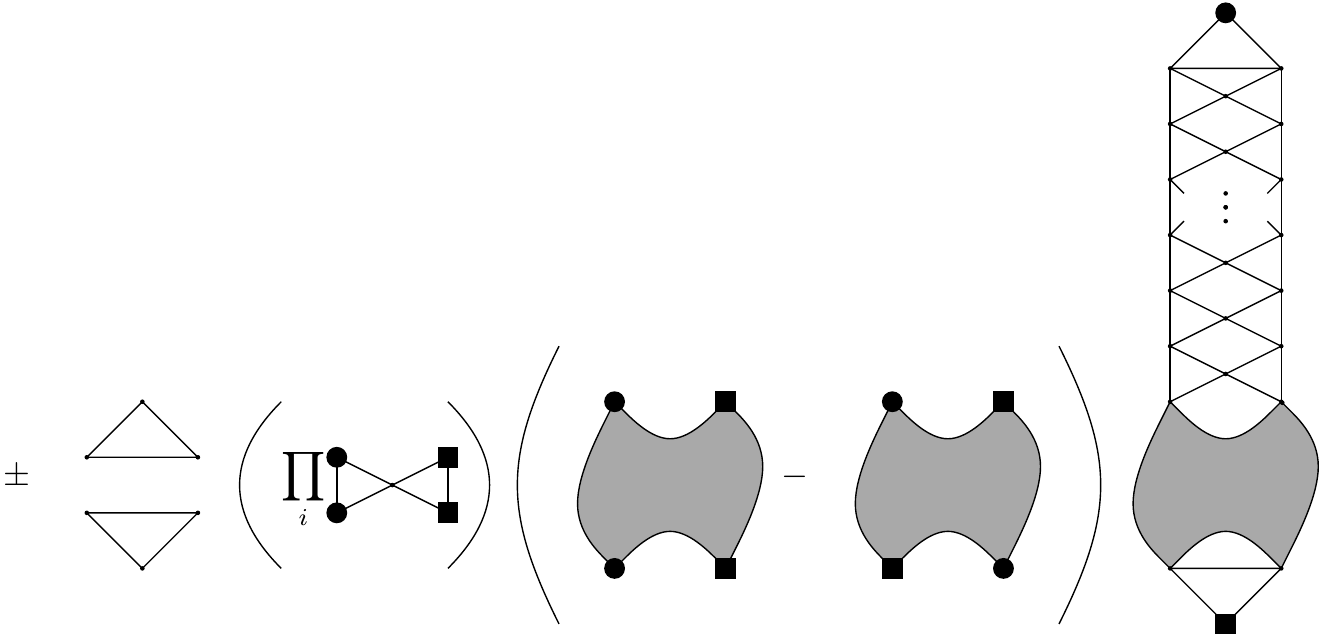}
 \caption{Diagrammatic representation of \eqref{eq the expression}.}\label{fig the expression}
\end{figure}

Note that due to the choice of $v$, one of the triangles $u_2$, $u_5$, $u_7$ and $u_3$, $u_8$, $u_9$ shares two vertices with $K$, while the other does not. Without loss of generality say that $u_2$, $u_5$, $u_7$ is the one that is adjacent to other hourglasses, if there are any. This is the one drawn pointing upwards in Figure~\ref{fig the expression}, and $u_2$ is the top vertex.

\subsection{The two triangles}

The plan of attack now is to reduce the edges in the two triangles $u_2$, $u_5$, $u_7$ and $u_3$, $u_8$, $u_9$. After the first such edge, we will need to pass to quadratic denominator reduction.
Note that the only factors in \eqref{eq the expression} containing edge variables from the triangle $u_2$, $u_5$, $u_7$ are $\Psi_{L[u_2, u_5, u_7]}$ and $\Phi_{L_2}^{\{u_2\}, \{u_3\}}$. Let the edge between $u_5$ and $u_7$ be $15$ and let $16$ and $17$ be the other two edges of the triangle $u_2$, $u_5$, $u_7$. Reduce 15 in the usual way to obtain
\begin{gather*}
(\alpha_{16}+\alpha_{17})\Phi_{L_2\backslash 15}^{\{u_2\}, \{u_3\}} - \Phi_{L_2/15}^{\{u_2\}, \{u_3\}}
\end{gather*}
times the factors not involving the triangle $u_2$, $u_5$, $u_7$. Expanding out $\alpha_{16}$ and $\alpha_{17}$ we get
\begin{gather*}
(\alpha_{16}+\alpha_{17})\Phi_{L_2\backslash 15}^{\{u_2\}, \{u_3\}} - \Phi_{L_2/15}^{\{u_2\}, \{u_3\}} \\ \qquad
{}=(\alpha_{16}+\alpha_{17})\Big(\alpha_{16}\Phi_{L_3}^{\{z_n\}, \{u_3\}} + \alpha_{17}\Phi_{L_3}^{\{x_n\}, \{u_3\}} + \alpha_{16}\alpha_{17}\Psi_{L_3}\Big) - \alpha_{16}\alpha_{17}\Psi_{L_4}
\\ \qquad
{}= \alpha_{16}^2\Phi_{L_3}^{\{z_n\}, \{u_3\}} + \alpha_{17}^2\Phi_{L_3}^{\{x_n\}, \{u_3\}} + 2\alpha_{16}\alpha_{17}\Phi_{L_3}^{\{x_n, z_n\}, \{u_3\}} + (\alpha_{16}+\alpha_{17})\alpha_{16}\alpha_{17}\Psi_{L_3}
\end{gather*}
times the factors not involving the triangle $u_3$, $u_5$, $u_7$, where
\begin{itemize}\itemsep=0pt
\item $n$ is the index of the hourglass adjacent to the triangle $u_2$, $u_5$, $u_7$ with the vertices $x_n$ and~$z_n$ being the same vertices as $u_5$ (which has also been contracted with the original $u_6$) and~$u_7$ respectively,
\item $16$ the edge from $x_n=u_5$ to $u_2$ and $17$ the edge from $z_n=u_7$ to $u_2$,
\item $L_3= (L_2-u_2)\backslash 15$ and $L_4 =(L_2-u_2)/ 15$.
\end{itemize}
This is illustrated diagrammatically in Figure~\ref{fig reduce 15}. Note that in these calculations we used that merging two vertices is equivalent to not having them merged but having them be in different parts of the vertex partition defining the spanning forest polynomial, and when there was originally only a single tree, there is no need to consider how the parts interact with any existing parts. We also used that if a vertex is not in a part of a partition for a spanning forest polynomial then we can sum over all possibilities for putting that vertex into a part.

\begin{figure}
\centering
 \includegraphics[width=\linewidth]{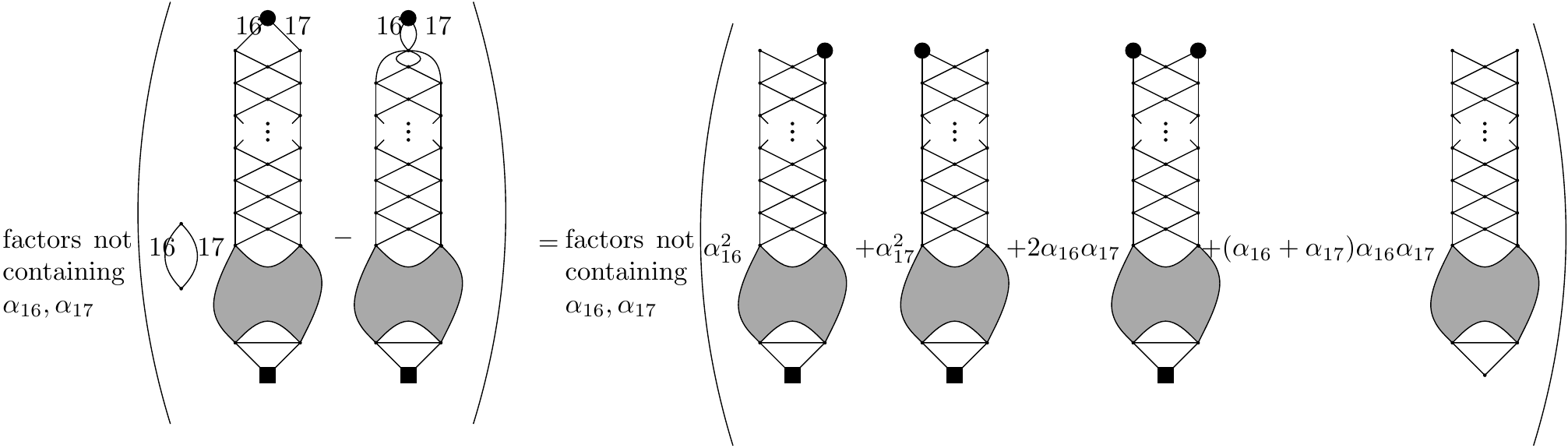}
 \caption{Diagrammatic rendition of the result of reducing edge 15.}\label{fig reduce 15}
\end{figure}

Now we need to move to quadratic denominator reduction, see Section~\ref{sec qdr}. Reducing $\alpha_{16}$ according to quadratic denominator reduction (\ref{case1}) we obtain
\begin{gather*}
\big(2\alpha_{17}\Phi_{L_3}^{\{x_n,z_n\}, \{u_3\}} + \alpha_{17}^2\Psi_{L_3} \Big)^2 - 4\Big(\Phi_{L_3}^{\{z_n\}, \{u_3\}} + \alpha_{17}\Psi_{L_3}\Big)\alpha_{17}^2\Phi_{L_3}^{\{x_n\}, \{u_3\}}
\end{gather*}
times the square of the factors not involving the triangle $u_3$, $u_5$, $u_7$. Note that $\alpha_{17}^2$ factors out of this expression and so following (\ref{case2})
quadratic denominator reduction of $\alpha_{17}$ gives
\begin{gather*}
4\Big( \big(\Phi_{L_3}^{\{x_n, z_n\}, \{u_3\}}\big)^2 - \Phi_{L_3}^{\{z_n\}, \{u_3\}}\Phi_{L_3}^{\{x_n\}, \{u_3\}}\Big)
\end{gather*}
times the square of the factors not involving the triangle $u_3$, $u_5$, $u_7$.

Next consider the $u_3$, $u_8$, $u_9$ triangle. Unfortunately we cannot simply use the same argument as above since we are not starting in conventional denominator reduction this time. This part of the argument is rather gruesome, but fortunately, it is the last bit of messy work before we get to the systematic part.

Let $18$ be the edge between $u_8$ and $u_9$, and let $19$ and $20$ be the edges from $u_3$ to $t_3$ and $t_4$ respectively. For the purposes of these edges, the factors we need to consider are the one given explicitly above and the factor for the $u_3$, $t_3$, $t_4$ triangle itself. Namely we have
\begin{gather*}
(\alpha_{18}+\alpha_{19}+\alpha_{20})^2\Big(\big(\Phi_{L_3}^{\{x_n, z_n\}, \{u_3\}}\big)^2 - \Phi_{L_3}^{\{z_n\}, \{u_3\}}\Phi_{L_3}^{\{x_n\}, \{u_3\}}\Big)
\end{gather*}
times the square of the factors not involving the triangles $u_3$, $u_5$, $u_7$ and $u_3$, $t_3$, $t_4$. We have included the constant $4$ among the suppressed factors as it also is simply carried along through the calculations that follow.

Reducing edge $18$ we obtain
\begin{gather*}
(\alpha_{19}+\alpha_{20})^2\Big(\big(\Phi_{L_3\backslash 18}^{\{x_n, z_n\}, \{u_3\}}\big)^2 - \Phi_{L_3\backslash 18}^{\{z_n\}, \{u_3\}}\Phi_{L_3\backslash 18}^{\{x_n\}, \{u_3\}}\Big)
\\ \qquad
{}- (\alpha_{19}\!+\!\alpha_{20})\Big(2\Phi_{L_3\backslash 18}^{\{x_n, z_n\}, \{u_3\}}\Phi_{L_3/18}^{\{x_n, z_n\}, \{u_3\}} \!- \!\Phi_{L_3\backslash 18}^{\{z_n\}, \{u_3\}}\Phi_{L_3/18}^{\{x_n\}, \{u_3\}} \!-\! \Phi_{L_3\backslash 18}^{\{x_n\}, \{u_3\}}\Phi_{L_3/18}^{\{z_n\}, \{u_3\}}\Big)
\\ \qquad
{}+ \Big(\big(\Phi_{L_3/18}^{\{x_n, z_n\}, \{u_3\}}\big)^2 - \Phi_{L_3/18}^{\{z_n\},\{u_3\}}\Phi_{L_3/18}^{\{x_n\}, \{u_3\}}\Big)
\end{gather*}
times the square of the factors not involving either triangle.
Expanding out all the $\alpha_{19}$ and $\alpha_{20}$ explicitly and freely using the observations on spanning forest polynomials used before, there are many nice cancellations and we obtain
\begin{gather*}
(\alpha_{19}+\alpha_{20})^2
\Big(\alpha_{19}^2\Big(\big(\Phi_{M_n}^{\{\{x_n, z_n\}, \{t_4\}}\big)^2 - \Phi_{M_n}^{\{z_n\}, \{t_4\}}\Phi_{M_n}^{\{x_n\}, \{t_4\}}\Big)
\\ \qquad
{} + \alpha_{20}^2\Big(\big(\Phi_{M_n}^{\{x_n, z_n\}, \{t_3\}}\big)^2 - \Phi_{M_n}^{\{z_n\}, \{t_3\}}\Phi_{M_n}^{\{x_n\}, \{t_3\}}\Big)
\\ \qquad
{}+ \alpha_{19}\alpha_{20}\Big(2\Phi_{M_n}^{\{x_n, z_n\}, \{t_3\}}\Phi_{M_n}^{\{x_n, z_n\}, \{t_4\}} - \Phi_{M_n}^{\{z_n\}, \{t_3\}}\Phi_{M_n}^{\{x_n\}, \{t_4\}}
 - \Phi_{M_n}^{\{z_n\}, \{t_4\}}\Phi_{M_n}^{\{x_n\}, \{t_3\}}
 \\ \qquad\hphantom{+ \alpha_{19}\alpha_{20}\Big(}
 + \Phi_{M_n}^{\{x_n\}, \{z_n\}}\Phi_{M_n}^{\{t_3\}, \{t_4\}}\Big)- \alpha_{19}\alpha_{20}(\alpha_{19}+\alpha_{20})\Psi_{M_n}\Phi^{\{x_n\}, \{z_n\}}\Big)
\end{gather*}
times the square of the factors not involving either triangle and where $M_n = (L_3-u_3)\backslash 18$ (the $n$ refers to the index of the last remaining hourglass). This calculation uses the fact that $\Phi_{M_n}^{\{z_n\},\{t_4\}} + \Phi_{M_n}^{\{x_n\}, \{t_4\}}-2\Phi_{M_n}^{\{x_n, z_n\}, \{t_4\}} = \Phi_{M_n}^{\{z_n\},\{t_3\}} + \Phi_{M_n}^{\{x_n\}, \{t_3\}}-2\Phi_{M_n}^{\{x_n, z_n\}, \{t_3\}} = \Phi_{M_n}^{\{x_n\}, \{z_n\}}$ which can be seen to be true by expanding all the spanning forest polynomials so that each of $x_n$, $z_n$, $t_3$, $t_4$ is in each partition.

Because of the factor of $(\alpha_{19}+\alpha_{20})^2$ we can proceed to reduce edge $19$ to obtain
\begin{gather*}
\alpha_{20}^2 \Big(\big(\Phi_{M_n}^{\{x_n, z_n\}, \{t_4\}}\big)^2 - \Phi_{M_n}^{\{z_n\}, \{t_4\}}\Phi_{M_n}^{\{x_n\}, \{t_4\}} - \alpha_{20}\Psi_{M_n}\Phi_{M_n}^{\{x_n\}, \{z_n\}}\Big)
\\ \qquad
{}- \alpha_{20} \Big(\alpha_{20}\big(2\Phi_{M_n}^{\{x_n, z_n\}, \{t_3\}}\Phi_{M_n}^{\{x_n z_n\}, \{t_4\}} - \Phi_{M_n}^{\{z_n\}, \{t_3\}}\Phi_{M_n}^{\{x_n\}, \{t_4\}}
 - \Phi_{M_n}^{\{z_n\},\{t_4\}}\Phi_{M_n}^{\{x_n\}, \{t_3\}}
 \\ \qquad\hphantom{- \alpha_{20} \Big(\alpha_{20}\big(}
{} + \Phi_{M_n}^{\{x_n\},\{z_n\}}\Phi_{M_n}^{\{t_3\}, \{t_4\}}\big)- \alpha_{20}^2\Psi_{M_n}\Phi^{\{x_n\}, \{z_n\}}\Big)
\\ \qquad
{}+ \alpha_{20}^2\Big(\big(\Phi_{M_n}^{\{x_n, z_n\}, \{t_3\}}\big)^2 - \Phi_{M_n}^{\{z_n\}, \{t_3\}}\Phi_{M_n}^{\{x_n\}, \{t_3\}}\Big)
\end{gather*}
times the square of the factors not involving either triangle. There is a factor of $\alpha_{20}^2$
in this expression, so we can reduce $\alpha_{20}$ to obtain
\begin{gather*}
\big(\Phi_{M_n}^{\{x_n, z_n\},\{t_4\}}\big)^2 - \Phi_{M_n}^{\{z_n\}, \{t_4\}}\Phi_{M_n}^{\{x_n\}, \{t_4\}} - 2 \Phi_{M_n}^{\{x_n, z_n\}, \{t_3\}}\Phi_{M_n}^{\{x_n, z_n\}, \{t_4\}} + \Phi_{M_n}^{\{z_n\}, \{t_3\}}\Phi_{M_n}^{\{x_n\}, \{t_4\}}
\\ \qquad
{}+ \Phi_{M_n}^{\{z_n\}, \{t_4\}}\Phi_{M_n}^{\{x_n\}, \{t_3\}} - \Phi_{M_n}^{\{x_n\}, \{z_n\}}\Phi_{M_n}^{\{t_3\}, \{t_4\}} + \big(\Phi_{M_n}^{\{x_n, z_n\}, \{t_3\}}\big)^2 - \Phi_{M_n}^{\{z_n\}, \{t_3\}}\Phi_{M_n}^{\{x_n\}, \{t_3\}}
\end{gather*}
times the square of the factors not involving either triangle. Expanding over all possibilities for assigning whichever of $x_n$, $z_n$, $t_3$, $t_4$ are not in the partition in each term, cancelling and then recollecting terms we can simplify the expression above to
\begin{gather}\label{eq done triangle}
 \Big(\Phi_{M_n}^{\{x_n, t_4\}, \{z_n, t_3\}} - \Phi_{M_n}^{\{x_n, t_3\}, \{z_n, t_4\}}\Big)^2 - \Phi_{M_n}^{\{x_n\}, \{z_n\}}\Phi_{M_n}^{\{t_3\}, \{t_4\}}
\end{gather}
times the square of the factors not involving either triangle. This expression is more symmetric than the notation makes it seem. Figure~\ref{fig done triangles} shows the symmetry better. Note that other than to choose the labelling of the vertices, we have not used the fact that all the remaining hourglasses are on one side of $K$. If we followed the same calculation beginning with $v$ in the middle of the hourglasses, then at this point we would have hourglasses on each side of $K$ making the expression nicely symmetric.

\begin{figure}
\centering
 \includegraphics[scale=.95]{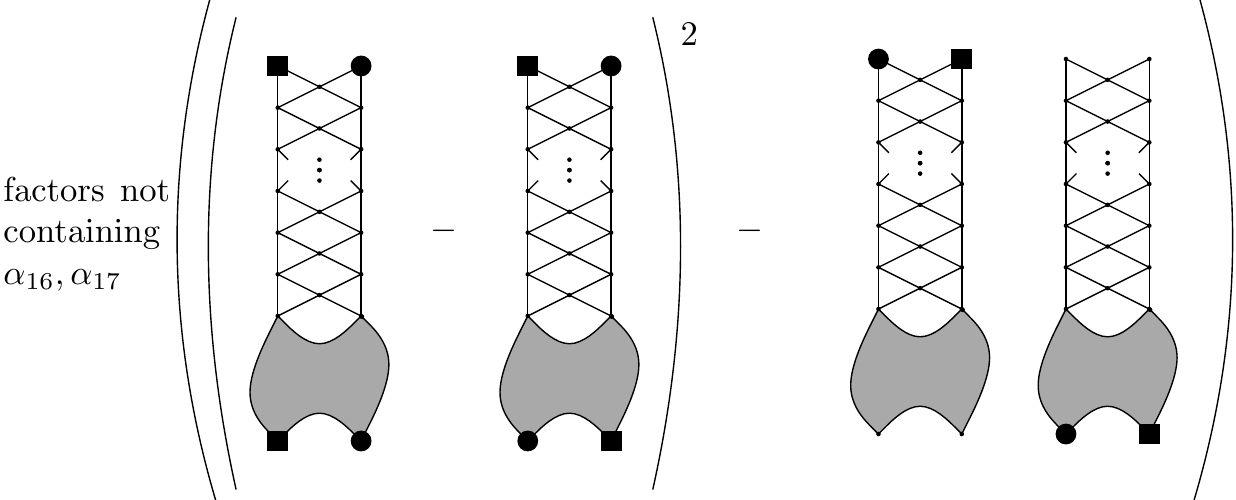}
 \caption{Diagrammatic representation of \eqref{eq done triangle}.}\label{fig done triangles}
\end{figure}

To proceed, we can apply a Dodgson identity to the last factor.
Let $M_n'$ be the graph obtained by adding an edge labelled $1$ joining $z_n$ and $x_n$ and an edge labelled $2$ joining $t_3$ and $t_4$ and let~$M_n'$.
Then \eqref{eq done triangle} can be written as $\big(\Psi_{M_n'}^{1,2}\big)^2 - \Psi^{1,1}_{2, M_n'}\Psi^{2,2}_{1, M_n'}$, so by the Dodgson identity (\ref{d12}) we find that
\eqref{eq done triangle} is also $-\Psi_{M_n', 12} \Psi_{M_n'}^{12,12}$.

Putting back in the factors we have been ignoring we get
\begin{gather}\label{eq after triangles}
 -4\bigg(\prod_{i}\Phi_{X_i}^{\{w,x\}, \{y,z\}} \bigg)^2\Big(\Phi_{K}^{\{t_1, t_3\}, \{t_2, t_4\}} - \Phi_{K}^{\{t_1, t_4\}, \{t_2, t_3\}}\Big)^2\Psi_{M_n', 12} \Psi_{M_n'}^{12,12}.
\end{gather}
For an illustration of this equation see Figure~\ref{fig after triangles}.

\begin{figure}
\centering
 \includegraphics[scale=.95]{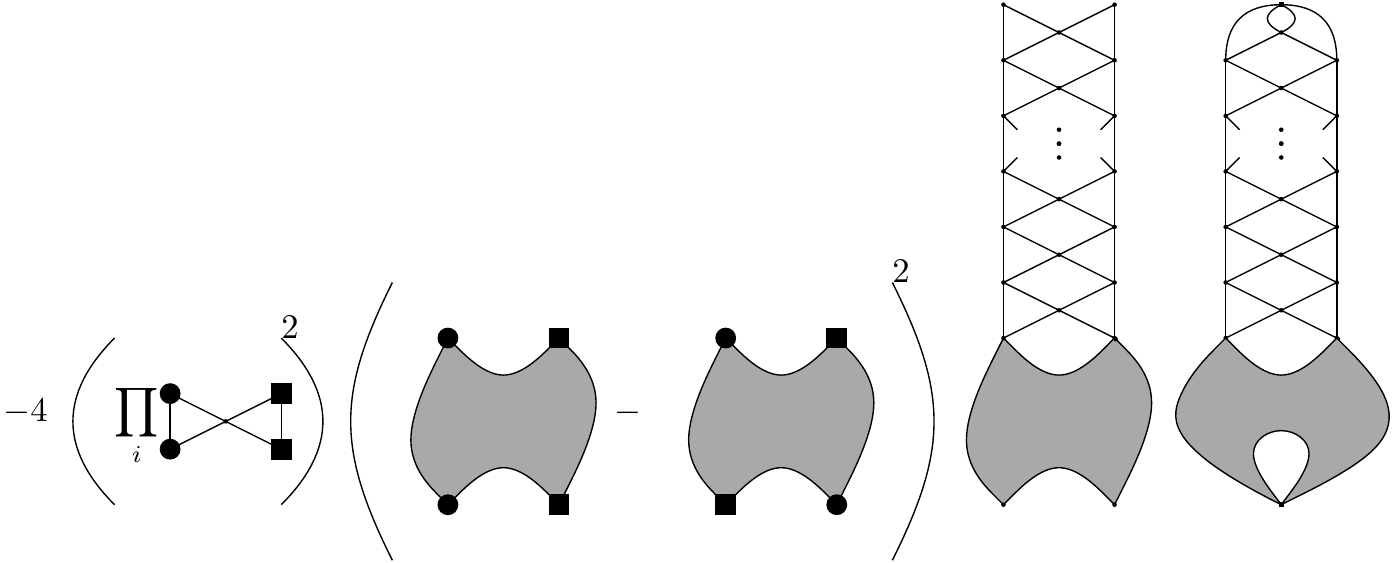}
 \caption{Diagrammatic representation of \eqref{eq after triangles}.}\label{fig after triangles}
\end{figure}

\subsection{Systematic hourglass reduction}

Notice that the two right hand factors of \eqref{eq after triangles} are the same except that $x_n$ and $z_n$, the top two vertices as illustrated, are identified or not
and likewise for $t_3$ and $t_4$, the bottom two vertices. This will be important, so it will be useful to have a compact notation for this; write
\begin{gather*}
\TA{n}\qquad \text{for} \quad \Psi_{M_n', 2}^{1,1},
\end{gather*}
that is the top vertices together ($T$) and the bottom two vertices apart ($A$), and similarly for
\begin{gather*}
\TT{n} = \Psi_{M_n', 12}, \qquad \AA{n} = \Psi_{M_n'}^{12,12}, \qquad \AT{n}= \Psi_{M_n', 2}^{1,1}.
\end{gather*}

The next order of business is to start reducing the top hourglass $X_n$. Recall the hourglass notation as in Figure~\ref{fig hourglass labels}.

First we reduce $a_n$. We are going to need to keep track to all the polynomials that come about from the remains of $X_n$ as they appear
in each term after reducing $a_n$. Writing generically for every hourglass for the moment, define
\begin{gather*}
 A = de,
 \qquad
 B = de(b+c)+bc(d+e+f),
 \\
 C = d+e+f,
 \qquad
 D = bd+(b+d)(e+f),
 \\
 E = ce+(c+e)(d+f),
 \qquad
 F = bd(c+e) + ce(b+d) + f(c+e)(b+d),
 \\
 G = (b+c)(d+e+f),
 \qquad
 H = bcd + (bc+bd+cd)(e+f),
 \\
 I = bce + (bc+be+ce)(d+f),
 \qquad
 J = f(bcd+bce+bde+cde).
\end{gather*}
These are the Kirchhoff polynomials and spanning forest polynomials corresponding to the graphs shown in Figure~\ref{fig catalogue}.

\begin{figure}
\centering
 \includegraphics[scale=.98]{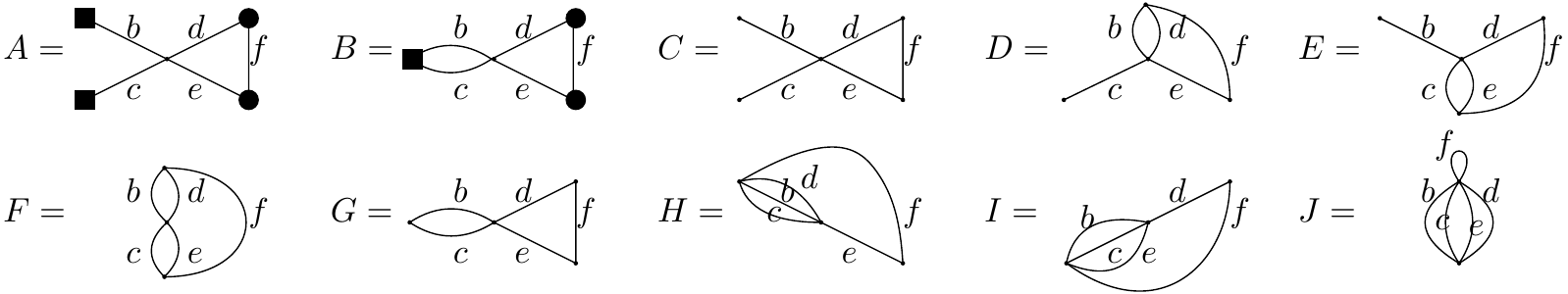}
 \caption{The hourglass pieces that occur.}\label{fig catalogue}
\end{figure}

Using \eqref{eq 2 cut} along with the notation above, reducing $a_n$ gives
\begin{gather*}
 -4\bigg(\prod_{i<n}\Phi_{X_i}^{\{w,x\}, \{y,z\}}\bigg)^2
 \Big(\Phi_{K}^{\{t_1, t_3\}, \{t_2, t_4\}} - \Phi_{K}^{\{t_1, t_4\}, \{t_2, t_3\}}\Big)^2
 \\ \qquad
{}\times \bigg(B_n^2 \bigg( C_nD_n \TA{n-1} \TT{n-1} + E_nD_n \AA{n-1} \TT{n-1} + C_nF_n \TA{n-1} \AT{n-1} + E_nF_n \AA{n-1} \AT{n-1}\bigg)
 \\ \qquad\hphantom{\times \bigg(}
{} - A_nB_n \bigg( G_nD_n \TA{n-1} \TT{n-1} + I_nD_n \AA{n-1} \TT{n-1}
 + G_nF_n \TA{n-1} \AT{n-1} + I_nF_n \AA{n-1} \AT{n-1}\bigg)
 \\ \qquad\hphantom{\times \bigg(}
 {} - A_nB_n \bigg( C_nH_n \TA{n-1} \TT{n-1} + E_nH_n \AA{n-1} \TT{n-1}
 + C_nJ_n \TA{n-1} \AT{n-1} + E_nJ_n \AA{n-1} \AT{n-1}\bigg)
 \\ \qquad\hphantom{\times \bigg(}
 + A_n^2 \le\bigg( G_nH_n \TA{n-1} \TT{n-1} + I_nH_n \AA{n-1} \TT{n-1}
 + G_nJ_n \TA{n-1} \AT{n-1} + I_nJ_n \AA{n-1} \AT{n-1}\bigg) \bigg).
\end{gather*}
This expression factors giving
\begin{gather*}
 -4\bigg(\prod_{i<n}\Phi_{X_i}^{\{w,x\}, \{y,z\}}\bigg)^2\Big(\Phi_{K}^{\{t_1, t_3\}, \{t_2, t_4\}} - \Phi_{K}^{\{t_1, t_4\}, \{t_2, t_3\}}\Big)^2
 \\ \qquad
{}\times \bigg(\TA{n-1}(B_nC_n-A_nG_n) + \AA{n-1}(B_nE_n - A_nI_n)\bigg)
 \\ \qquad
{}\times \bigg(\TT{n-1}(B_nD_n - A_nH_n) + \AT{n-1}(B_nF_n - A_nJ_n)\bigg).
\end{gather*}
Plugging in the expressions for the polynomials this simplifies further to
\begin{gather}
 -4\bigg(\prod_{i<n}\Phi_{X_i}^{\{w,x\}, \{y,z\}}\bigg)^2\Big(\Phi_{K}^{\{t_1, t_3\}, \{t_2, t_4\}} - \Phi_{K}^{\{t_1, t_4\}, \{t_2, t_3\}}\Big)^2
 (d_n+e_n+f_n)c_nZ_n\nonumber
\\ \qquad
 {}\times\bigg(\TA{n-1}b_n(d_n+e_n+f_n)+\AA{n-1}Y_n\bigg)\bigg(\TT{n-1}b_n(d_n+e_n+f_n) + \AT{n-1}Y_n\bigg),
\label{eq will return}
\end{gather}
where
\begin{gather*}
 Y = bcd+bce+bde+cde+bcf+bef,
 \\
 Z = bcd+bce+bde+cde+bcf+cdf.
\end{gather*}

Next we reduce $a_{n-1}$, yielding
\begin{gather*}
 -4\bigg(\prod_{i<n-1}\Phi_{X_i}^{\{w,x\}, \{y,z\}}\bigg)^2\Big(\Phi_{K}^{\{t_1, t_3\}, \{t_2, t_4\}} - \Phi_{K}^{\{t_1, t_4\}, \{t_2, t_3\}}\Big)^2 (d_n+e_n+f_n)c_nZ_n
 \\ \qquad
 {}\times \bigg(\TA{n-2}((d_n+e_n+f_n)b_n(B_{n-1}D_{n-1}-A_{n-1}H_{n-1}) + Y_n(B_{n-1}C_{n-1}-A_{n-1}G_{n-1}))
 \\ \qquad\hphantom{\times \bigg(}
 {} + \AA{n-2}((d_n+e_n+f_n)b_n(B_{n-1}F_{n-1}-A_{n-1}J_{n-1}) + Y_n(B_{n-1}E_{n-1}-A_{n-1}I_{n-1}))\bigg)
 \\ \qquad
 {}\times \bigg(\TT{n-2}((d_n+e_n+f_n)b_n(B_{n-1}D_{n-1}-A_{n-1}H_{n-1}) + Y_n(B_{n-1}C_{n-1}-A_{n-1}G_{n-1}))
 \\ \qquad \hphantom{\times \bigg(}
 {}+ \AT{n-2}((d_n+e_n+f_n)b_n(B_{n-1}F_{n-1}-A_{n-1}J_{n-1}) + Y_n(B_{n-1}E_{n-1}-A_{n-1}I_{n-1}))\bigg).
\end{gather*}
Substituting in the expressions for the polynomials something special happens; a square factor appears:
\begin{gather*}
 -4\bigg(\prod_{i<n-1}\Phi_{X_i}^{\{w,x\}, \{y,z\}}\bigg)^2\Big(\Phi_{K}^{\{t_1, t_3\}, \{t_2, t_4\}} - \Phi_{K}^{\{t_1, t_4\}, \{t_2, t_3\}}\Big)^2 (d_n+e_n+f_n)c_nZ_n
 \\ \qquad
 {}\times \big((d_n+e_n+f_n)b_nZ_{n-1} + Y_nc_{n-1}(d_{n-1}+e_{n-1}+f_{n-1})\big)^2
 \\ \qquad
 {}\times \bigg(\TA{n-2}(d_{n-1}\!+e_{n-1}\!+\!f_{n-1})b_{n-1}\! + \AA{n-2}Y_{n-1}\!\bigg)\bigg(\TT{n-2}(d_{n-1}\!+e_{n-1}\!+f_{n-1})b_{n-1}\!+ \AT{n-2}Y_{n-1}\!\bigg).
\end{gather*}

We have a square term $T_{n-1,n}^2$, where $T_{n-1,n} =(d_n+e_n+f_n)b_nZ_{n-1} + Y_nc_{n-1}(d_{n-1}+e_{n-1}\allowbreak+f_{n-1})$. The factor $T_{n-1,n}$ involves edge variables from both of the top two hourglasses. It is the spanning forest polynomial illustrated in Figure~\ref{fig bihourglass factor}. We use this factor to reduce the variables of the hourglass $X_n$ (in any order).
This part of the calculation is routine and can be done by hand or rigorously on a computer (reductions are, e.g., implemented in~\cite{Shlog}), obtaining
\begin{gather*}
(d_n+e_n+f_n)c_nZ_nT_{n-1,n}^2\rightarrow c_{n-1}(d_{n-1}+e_{n-1}+f_{n-1})Z_{n-1}
\end{gather*}
which exactly leads to \eqref{eq will return} with $n-1$ in place of $n$.

\begin{figure}
\centering
 \includegraphics{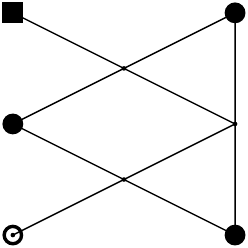}
 \caption{Spanning forest polynomial giving the factor $T_{n-1,n}$.}\label{fig bihourglass factor}
\end{figure}

Inductively, we can reduce until only pieces from one hourglass $X_1$ remain. At this step the quadratic denominator reduction will give
\begin{gather*}
 -4 \Big(\Phi_{K}^{\{t_1, t_3\}, \{t_2, t_4\}} - \Phi_{K}^{\{t_1, t_4\}, \{t_2, t_3\}}\Big)^2 c_{1}(d_{1}+e_{1}+f_{1})Z_{1}
 \\ \qquad
 {}\times\bigg(\TA{0}(d_{1}+e_{1}+f_{1})b_{1} + \AA{0}Y_{1}\bigg)\bigg(\TT{0}(d_{1}+e_{1}+f_{1})b_{1}+ \AT{0}Y_{1}\bigg).
\end{gather*}

Next we want to rewrite the parts involving the kernel in terms of Dodgson polynomials rather than spanning forest polynomials or the apart-together notation (which no longer
has any hourglasses in it). Using $K'$ as in the statement of the theorem this gives (see Figure \ref{fig Kprime})
\begin{gather}
-4c_1(d_1+e_1+f_1)Z_1\big(\Psi_{K'}^{1,2}\big)^2\big(\Psi_{K', 1}^{2,2}(d_1+e_1+f_1)b_1+\Psi_{K'}^{12,12}Y_1\big)\nonumber
\\ \qquad
{}\times\big(\Psi_{K', 12}(d_1+e_1+f_1)b_1+ \Psi_{K', 2}^{1,1}Y_1\big).
\label{finaleq}
\end{gather}

\subsection{The endgame}\label{sec endgame}

Because the expression (\ref{finaleq}) vanishes modulo 2 we get the result for $q=2$ from the last statement in Theorem~\ref{thmqdr}.

We now restrict ourselves to odd prime powers and show that the last hourglass can be eliminated. With (\ref{finaleq}) we achieved the following situation: we have two sets of variables,
the five variables $b_1,c_1,d_1,e_1,f_1$ from the hourglass $X_1$ and some variables $\alpha_i$ in the Dodgson polynomials with the kernel $K'$.
The expression (\ref{finaleq}) does not depend on the two variables $\alpha_1$ and $\alpha_2$ which are associated to the extra edges 1 and 2 in $K'$.
Let $d$ be the degree of $\Psi_{K'}^{12,12}$. Then $\Psi_{K', 1}^{2,2}$, $\Psi_{K'}^{1,2}$, and $\Psi_{K', 2}^{1,1}$ have degree $d+1$,
while $\Psi_{K', 12}$ has degree $d+2$ (see, e.g.,~\cite{Sc2}). The total degree of (\ref{finaleq}) is $4d+14$ which equals twice the total number of its variables.

Quadratic denominator reduction stops at (\ref{finaleq}). To obtain further reductions we use a scaling technique
which was first used in~\cite{SFq} and later adopted in~\cite{K3} to exhibit a K3 structure in $\phi^4$ theory at loop order eight. Considering (\ref{finaleq}) as a denominator
of an integrand it is clear that the variables separate under a scaling transformation of all $\alpha_i$ by $S=Y_1/[b_1c_1(d_1+e_1+f_1)]$. Because $S$
has total degree zero, homogeneity of (\ref{finaleq}) is preserved.

For Legendre sums over finite fields we need the following argument. If $S\in\FF_q^\times$ the $\alpha_i$-variables in the Legendre sum can be multiplied by $S$ yielding
\begin{gather*}
-b_1(d_1+e_1+f_1)^2Y_1Z_1\big(\Psi_{K'}^{1,2}\big)^2\big(\Psi_{K', 1}^{2,2}+c_1\Psi_{K'}^{12,12}\big) \big(\Psi_{K', 12}+c_1\Psi_{K', 2}^{1,1}\big)(2S^{2d+2})^2.
\end{gather*}
The last factor is a non-zero square which can be dropped from the Legendre sum. This suggests that
\begin{gather}
(\text{expression (\ref{finaleq}}))_q\nonumber
\\ \qquad
{}=\Big({-}b_1(d_1+e_1+f_1)^2Y_1Z_1\big(\Psi_{K'}^{1,2}\big)^2\big(\Psi_{K', 1}^{2,2}+c_1\Psi_{K'}^{12,12}\big)\big(\Psi_{K', 12}+c_1\Psi_{K', 2}^{1,1}\big)\Big)_q.
\label{endgame}
\end{gather}

In $\FF_q$, we cannot ignore the singular locus of the scaling transformation. We need the following lemma.
\begin{Lemma}
Let $P$ be a homogeneous polynomial of odd degree and let $q$ be an odd prime power. Then $(P)_q=0$.
\end{Lemma}
\begin{proof}
Because $q$ is odd there exists an $x\in\FF_q^\times$ which is not a square (half the elements in $\FF_q^\times$ are non-squares). Scaling all variables by $x$ gives
\begin{gather*}
(P)_q=\big(Px^D\big)_q=(P)_q\genfrac(){}{}{x}{q}^D=-(P)_q,
\end{gather*}
where we used that the degree $D$ of $P$ is odd.
\end{proof}
We observe that in any situation where $S$ is singular (i.e., some of the $Y_1$, $b_1$, $c_1$, $d_1+e_1+f_1$ are zero) both expressions~-- (\ref{finaleq}) and the polynomial on the
right hand side of (\ref{endgame})~-- are either zero or the product of two factors in separate variables which are homogeneous of odd degree. The validity of (\ref{endgame}) follows from the
Lemma with an inclusion-exclusion argument.

The term on the right hand side of (\ref{endgame}) is homogeneous of degree $4d+14$ which equals twice the number of variables. We may use quadratic denominator reduction
in the variables $f_1$, $e_1$, $d_1$, $b_1$ (in this sequence) yielding (use, e.g.,~\cite{Shlog})
\begin{gather*}
(\text{expression (\ref{finaleq}}))_q\equiv\Big(c_1\big(\Psi_{K'}^{1,2}\big)^2\big(\Psi_{K', 1}^{2,2}+c_1\Psi_{K'}^{12,12}\big)\big(\Psi_{K', 12}+c_1\Psi_{K', 2}^{1,1}\big)\Big)_q\mod q.
\end{gather*}
By contraction-deletion (\ref{cd}) the polynomial on the right hand side is $\alpha_1\big(\Psi_{K'}^{1,2}\big)^2\Psi_{K'}^{2,2}\Psi_{K', 2}$, where we renamed $c_1$ to $\alpha_1$.

Theorem~\ref{mainthm} follows from Theorem~\ref{thmqdr} because the number of reduced variables is odd.

\section{Kernels}\label{sec kernels}
In this section we study kernels which lead to 4-regular hourglass chains. This implies that the kernel $K$ is internally 4-regular while every external vertex has two incident edges
(see~Figure~\ref{Fig:cases}).

\subsection{Trivial kernels}
Periods which admit a 3-vertex split are products (see Figure \ref{fig product}).
Assuming the completion conjecture, their $c_2$ invariants vanish. We hence skip kernels with a 3-vertex split.
If a kernel has a double triangle (see Figure \ref{fig double triangle reduction}), the $c_2$ invariant is equal to the $c_2$ of a smaller graph,
where the double triangle is reduced~\cite{BSYc2, Scensus}.
We also exclude these cases in $K$ because their $c_2$s are found in smaller kernels.
Moreover, we ignore kernels which have an external hourglass in such a way that it adds to the chain (with the exception that $K$ is an hourglass).

Assume a kernel $K$ with at least two internal vertices has two or more edges between external vertices.
Then, every hourglass chain $L\in \mathcal{G}_K$ splits if one cuts the four or less edges of $K$ which have exactly one external vertex.
The chain $L$ has a subdivergence, its period diverges and the~$c_2$ vanishes~\cite{BSYc2}. The same holds true if $K$ has a non-trivial internal four edge cut.

So, for kernels $K$ with $\geq2$ internal vertices we restrict ourselves to the case that $K$ has at most one edge between external vertices.
If such an edge $e$ shares its vertices with edge~1 or~2 in~$K'$, then any $L\in \mathcal{G}_K$ has a double-triangle.
After two double-triangle reductions we are left with a graph which has a three vertex split and the $c_2$ (conjecturally) vanishes.
We are effectively left with the case that $e$ joins the vertices of the hourglasses from opposite ends of the chain.
By Theorem~\ref{mainthm} M\"obius twists can be ignored as they lead to equal $c_2$s.
The case of one edge between external vertices of $K$ thus reduces to a single setup (which becomes relevant for kernels with $\geq6$ internal vertices).

If $K$ has no edge between external vertices, there exist three potentially distinct cases how to glue the kernel $K$ into the hourglass chain (corresponding to the 2,2 set
partitions of the external vertices).

\subsection{Small kernels}
We generated all effectively different kernels with up to ten internal vertices. We use the fact that every kernel $K$ can be made 4-regular by
adding a square to the external edges. Non-trivial kernels with at least one internal vertex can be found in 4-regular graphs which are internally six-connected and do not have
a three vertex cut. Such 4-regular graphs are called irreducible primitive in~\cite{Scensus}. From opening these graphs along all their squares we obtain the number of effectively
different kernels given in Table \ref{tab:numberkernels} (graphs were generated with nauty~\cite{NAU}). See Figure~\ref{Fig:cases} for the cases with at most five internal vertices.

\begin{figure}
\centering
\includegraphics{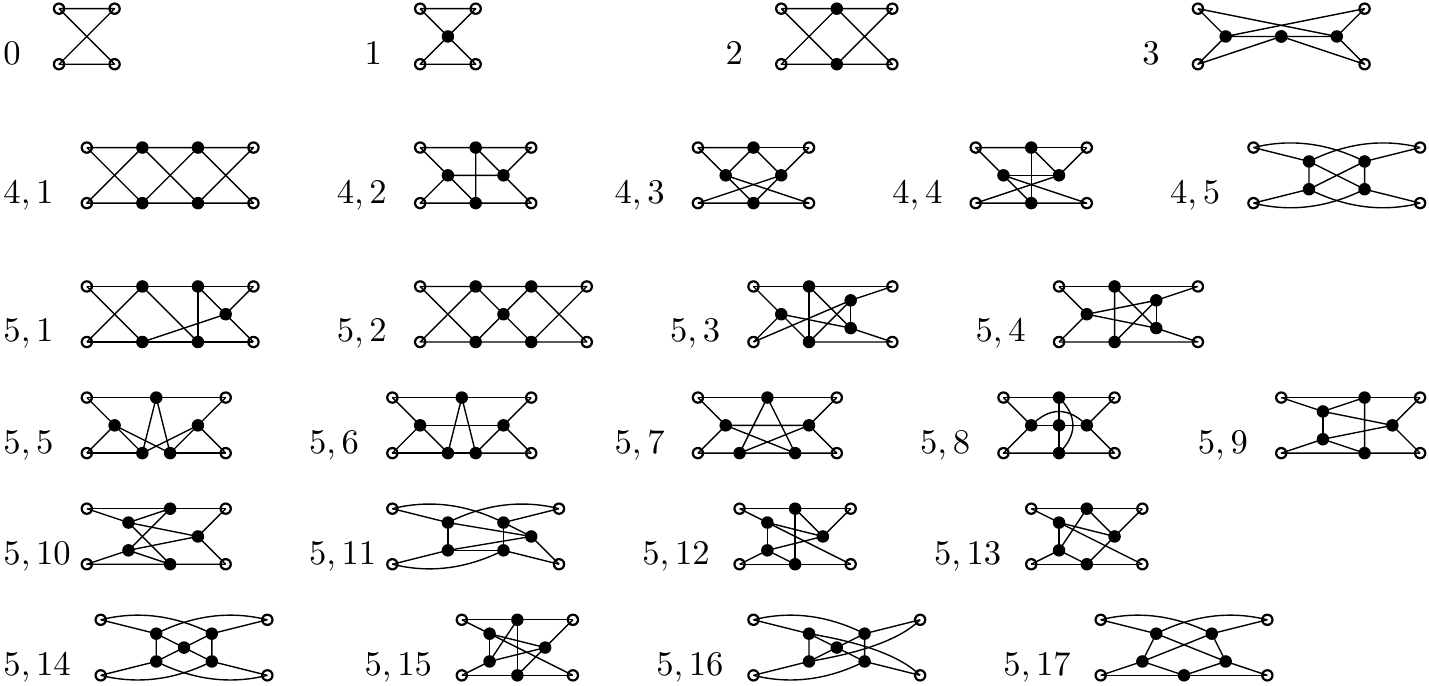}
\caption{The kernels with up to five internal vertices. Here, edges 1 and 2 in $K'$ would be vertical edges on either sides of the depicted kernels $K$.}\label{Fig:cases}
\end{figure}

\begin{table}\renewcommand{\arraystretch}{1.2}
\caption{The number of effectively different kernels in $\phi^4$ theory increases rapidly with the number of internal vertices.}\label{tab:numberkernels}
\begin{center}
\begin{tabular}{l|ccccccccccc}
\hline
\# internal vertices&0&1&2&3&4&5&6&7&8&9&10
\\ \hline
\# different kernels&1&1&1&1&5&17&78&497&3882&33587&316860
\\
\hline
\end{tabular}
\end{center}
\end{table}

The connection between $c_2$ invariants and geometries is described in~\cite{BSmod,Sc2}. Here, we investigated all kernels with at most six internal vertices.
For kernels with up to five internal vertices the $c_2$ invariants are listed in Table \ref{Tab:c2}. Note that most results are not proven but obtained by identifying finite $c_2$ prefixes
(indexed by primes). In general it is convenient to restrict prefix calculations to primes because (1) computations are simpler and faster, (2) it is assumed that prime prefixes determine
the geometry, see~\cite[Conjecture~2]{Sc2}, and (3) modularity only uses primes, see~\cite[Definition~21]{BSmod}.
All identified modular $c_2$s are confirmed up to prime~29 using the Maple package {\tt HyperlogProcedures} by the first author~\cite{Shlog}.
Computations for the primes 31 and 37 are ongoing.
\begin{table}[t!]\renewcommand{\arraystretch}{1.1}
\caption{The $c_2$ invariants for kernels with up to five internal vertices. The first column in the table refers to the number of internal vertices and the label in Figure~\ref{Fig:cases}.
Unidentified sequences in the second column are specified by their prime prefix $\big({-}c_2^{(p)}\mod p\big)_{p=2,3,5,7,11,\dots}$ (note that modularity, e.g., applies to the negative $c_2$).\vspace{1ex}}\label{Tab:c2}
\centering
\begin{tabular}{l|l}
\hline
\multicolumn{1}{c|}{kernel} &\multicolumn{1}{c}{$c_2$ invariant}
\\
\hline\hline
0&Legendre symbol $(-4/q)$
\\\hline
1&Legendre symbol $(4/q)$
\\\hline
2&Legendre symbol $(4/q)$
\\\hline
3&modular form of weight 4 and level 8
\\\hline
4,1&Legendre symbol $(-4/q)$
\\
4,2&unidentified sequence $0,2,3,2,3,8,15,9,6,27,11,32\dots$
\\
4,3&Legendre symbol $(4/q)$
\\
4,4&unidentified sequence $0,0,0,0,0,7,16,0,0,22,0,19,\dots$
\\
4,5&unidentified sequence $0,1,3,5,8,8,15,10,17,27,20,32\dots$
\\\hline
5,1&modular form of weight 4 and level 16
\\
5,2&Legendre symbol $(4/q)$
\\
5,3&modular form of weight 9 and level 4
\\
5,4&unidentified sequence $0,1, 1, 1, 2, 4,\dots$
\\
5,5&unidentified sequence $0,1, 1, 0, 2, 0, 3,\dots$
\\
5,6&unidentified sequence $0,1, 4, 2, 0, 3, 1,\dots$
\\
5,7&unidentified sequence $0,2, 4, 2, 8, 4,\dots$
\\
5,8&unidentified sequence $0,2, 0, 6, 10, 9, 10,\dots$
\\
5,9&unidentified sequence $0,2, 4, 5, 0, 3, 1, 2,\dots$
\\
5,10&unidentified sequence $0,1, 1, 1, 1, 7, 6, 17, 2,\dots$
\\
5,11&unidentified sequence $0,1, 1, 6, 1, 11, 2,\dots$
\\
5,12&unidentified sequence $0,1, 4, 5, 4, 4, 11,\dots$
\\
5,13&unidentified sequence $0,0, 1, 5, 1, 3, 16,\dots$
\\
5,14&modular form of weight 6 and level 4
\\
5,15&unidentified sequence $0,0, 0, 4, 5, 10, 3,\dots$
\\
5,16&unidentified sequence $0,0,3,0,0,3,1,0,0,25,\dots$
\\
5,17&modular form of weight 6 and level 4
\\
\hline
\end{tabular}
\end{table}

Full reductions were possible for the Legendre symbol $(-4/q)$ in kernel 4,1 and for the Legendre symbol $(4/q)$ in kernel 5,2 and in a kernel with six internal vertices.
For these hourglass families of $\phi^4$ ancestors the $c_2$ is proved (except for non-trivial even prime powers).

The modular form [9,4] in kernel 5,3 was not found in the $c_2$ invariants of $\phi^4$ graphs of loop order $\leq12$ (see~\cite{Sc2}).
It is the first form of weight 9 that has been found in $\phi^4$ theory. Note that 4 is the lowest level of all forms of weight 9 which
fits into the picture that forms in $\phi^4$ have very low level. No new modular forms were found in hourglass chains of kernels with six internal vertices.

Assuming the completion conjecture we could show that no weight 2 modular form of level $\leq1000$ (corresponding to point-counts of curves) exists
in hourglass chains of kernels with at most six internal vertices. This result provides substantial extra support for the no-curve conjecture in~\cite{BSmod,Sc2}.

The analysis of kernels with seven (or more) internal vertices requires significantly more computing power. We did not pursue this here.

\subsection*{Acknowledgements}
Both authors are deeply indebted to Dirk Kreimer for many years of encouragement and support.
Oliver Schnetz is supported by DFG grant SCHN~1240. Karen Yeats is supported by an NSERC Discovery grant and by the Canada Research Chairs program; during some of this work she was visiting Germany as a Humboldt fellow.

\pdfbookmark[1]{References}{ref}
\LastPageEnding


\begin{thebibliography}{99}
\footnotesize\itemsep=0pt

\bibitem{BEK}
Bloch S., Esnault H., Kreimer D., On motives associated to graph polynomials,
 \href{https://doi.org/10.1007/s00220-006-0040-2}{\textit{Comm. Math. Phys.}} \textbf{267} (2006), 181--225,
 \href{https://arxiv.org/abs/math.AG/0510011}{arXiv:math.AG/0510011}.

\bibitem{gfe}
Borinsky M., Schnetz O., Graphical functions in even dimensions,
 \href{https://arxiv.org/abs/2105.05015}{arXiv:2105.05015}.

\bibitem{BK}
Broadhurst D.J., Kreimer D., Knots and numbers in {$\phi^4$} theory to {$7$}
 loops and beyond, \href{https://doi.org/10.1142/S012918319500037X}{\textit{Internat.~J. Modern Phys.~C}} \textbf{6} (1995),
 519--524, \href{https://arxiv.org/abs/hep-ph/9504352}{arXiv:hep-ph/9504352}.

\bibitem{Brbig}
Brown F., On the periods of some Feynman integrals, \href{https://arxiv.org/abs/0910.0114}{arXiv:0910.0114}.

\bibitem{BrH1}
Brown F., The massless higher-loop two-point function, \href{https://doi.org/10.1007/s00220-009-0740-5}{\textit{Comm. Math.
 Phys.}} \textbf{287} (2009), 925--958, \href{https://arxiv.org/abs/0804.1660}{arXiv:0804.1660}.

\bibitem{Bcoact1}
Brown F., Feynman amplitudes, coaction principle, and cosmic {G}alois group,
 \href{https://doi.org/10.4310/CNTP.2017.v11.n3.a1}{\textit{Commun. Number Theory Phys.}} \textbf{11} (2017), 453--556,
 \href{https://arxiv.org/abs/1512.06409}{arXiv:1512.06409}.

\bibitem{Bcoact2}
Brown F., Notes on motivic periods, \href{https://doi.org/10.4310/CNTP.2017.v11.n3.a2}{\textit{Commun. Number Theory Phys.}}
 \textbf{11} (2017), 557--655, \href{https://arxiv.org/abs/1512.06410}{arXiv:1512.06410}.

\bibitem{BD}
Brown F., Doryn D., Framings for graph hypersurfaces, \href{https://arxiv.org/abs/1301.3056}{arXiv:1301.3056}.

\bibitem{K3}
Brown F., Schnetz O., A {K}3 in {$\phi^4$}, \href{https://doi.org/10.1215/00127094-1644201}{\textit{Duke Math.~J.}} \textbf{161}
 (2012), 1817--1862, \href{https://arxiv.org/abs/1006.4064}{arXiv:1006.4064}.

\bibitem{BSmod}
Brown F., Schnetz O., Modular forms in quantum field theory, \href{https://doi.org/10.4310/CNTP.2013.v7.n2.a3}{\textit{Commun.
 Number Theory Phys.}} \textbf{7} (2013), 293--325, \href{https://arxiv.org/abs/1304.5342}{arXiv:1304.5342}.

\bibitem{BSYc2}
Brown F., Schnetz O., Yeats K., Properties of {$c_2$} invariants of {F}eynman
 graphs, \href{https://doi.org/10.4310/ATMP.2014.v18.n2.a2}{\textit{Adv. Theor. Math. Phys.}} \textbf{18} (2014), 323--362,
 \href{https://arxiv.org/abs/1203.0188}{arXiv:1203.0188}.

\bibitem{BrY}
Brown F., Yeats K., Spanning forest polynomials and the transcendental weight
 of {F}eynman graphs, \href{https://doi.org/10.1007/s00220-010-1145-1}{\textit{Comm. Math. Phys.}} \textbf{301} (2011),
 357--382, \href{https://arxiv.org/abs/0910.5429}{arXiv:0910.5429}.

\bibitem{CYgrid}
Chorney W., Yeats K., {$c_2$} invariants of recursive families of graphs,
 \href{https://doi.org/10.4171/AIHPD/72}{\textit{Ann. Inst. Henri Poincar\'e~D}} \textbf{6} (2019), 289--311,
 \href{https://arxiv.org/abs/1701.01208}{arXiv:1701.01208}.

\bibitem{Denham}
Denham G., Schulze M., Walther U., Matroid connectivity and singularities of
 configuration hypersurfaces, \href{https://doi.org/10.1007/s11005-020-01352-3}{\textit{Lett. Math. Phys.}} \textbf{111} (2021),
 11, 67~pages, \href{https://arxiv.org/abs/1902.06507}{arXiv:1902.06507}.

\bibitem{HSSYc2}
Hu S., Schnetz O., Shaw J., Yeats K., Further investigations into the graph
 theory of $\phi^4$-periods and the $c_2$ invariant, \textit{Ann. Inst. Henri
 Poincar\'e~D}, {t}o appear, \href{https://arxiv.org/abs/1812.08751}{arXiv:1812.08751}.

\bibitem{IZ}
Itzykson C., Zuber J.B., Quantum field theory, \textit{International Series in Pure and
 Applied Physics}, McGraw-Hill International Book Co., New York, 1980.

\bibitem{KP}
Kompaniets M.V., Panzer E., Minimally subtracted six-loop renormalization of
 {$O(n)$}-symmetric {$\phi^4$} theory and critical exponents, \href{https://doi.org/10.1103/physrevd.96.036016}{\textit{Phys.
 Rev.~D}} \textbf{96} (2017), 036016, 26~pages, \href{https://arxiv.org/abs/1705.06483}{arXiv:1705.06483}.

\bibitem{Lef}
Lefschetz S., On the fixed point formula, \href{https://doi.org/10.2307/1968838}{\textit{Ann. of Math.}} \textbf{38}
 (1937), 819--822.

\bibitem{NAU}
McKay B.D., Piperno A., Practical graph isomorphism,~{II}, \href{https://doi.org/10.1016/j.jsc.2013.09.003}{\textit{J.~Symbolic
 Comput.}} \textbf{60} (2014), 94--112, \href{https://arxiv.org/abs/1301.1493}{arXiv:1301.1493}.

\bibitem{Panzer:HyperInt}
Panzer E., Algorithms for the symbolic integration of hyperlogarithms with
 applications to {F}eynman integrals, \href{https://doi.org/10.1016/j.cpc.2014.10.019}{\textit{Computer Phys. Comm.}}
 \textbf{188} (2015), 148--166, \href{https://arxiv.org/abs/1403.3385}{arXiv:1403.3385}.

\bibitem{PScoaction}
Panzer E., Schnetz O., The {G}alois coaction on {$\phi^4$} periods,
 \href{https://doi.org/10.4310/CNTP.2017.v11.n3.a3}{\textit{Commun. Number Theory Phys.}} \textbf{11} (2017), 657--705,
 \href{https://arxiv.org/abs/1603.04289}{arXiv:1603.04289}.

\bibitem{Patt}
Patterson E., On the singular structure of graph hypersurfaces, \href{https://doi.org/10.4310/CNTP.2010.v4.n4.a3}{\textit{Commun.
 Number Theory Phys.}} \textbf{4} (2010), 659--708, \href{https://arxiv.org/abs/1004.5166}{arXiv:1004.5166}.

\bibitem{intmot}
Rella C., An introduction to motivic {F}eynman integrals, \href{https://doi.org/10.3842/SIGMA.2021.032}{\textit{SIGMA}}
 \textbf{17} (2021), 032, 56~pages, \href{https://arxiv.org/abs/2009.00426}{arXiv:2009.00426}.

\bibitem{Scensus}
Schnetz O., Quantum periods: a census of {$\phi^4$}-transcendentals,
 \href{https://doi.org/10.4310/CNTP.2010.v4.n1.a1}{\textit{Commun. Number Theory Phys.}} \textbf{4} (2010), 1--47,
 \href{https://arxiv.org/abs/0801.2856}{arXiv:0801.2856}.

\bibitem{SFq}
Schnetz O., Quantum field theory over {$\mathbb F_q$}, \href{https://doi.org/10.37236/589}{\textit{Electron.~J.
 Combin.}} \textbf{18} (2011), 102, 23~pages, \href{https://arxiv.org/abs/0909.0905}{arXiv:0909.0905}.

\bibitem{Snumfunct}
Schnetz O., Numbers and functions in quantum field theory, \href{https://doi.org/10.1103/physrevd.97.085018}{\textit{Phys.
 Rev.~D}} \textbf{97} (2018), 085018, 20~pages, \href{https://arxiv.org/abs/1606.08598}{arXiv:1606.08598}.

\bibitem{Sc2}
Schnetz O., Geometries in perturbative quantum field theory, \href{https://doi.org/10.4310/CNTP.2021.v15.n4.a2}{\textit{Commun.
 Number Theory Phys.}} \textbf{15} (2021), 743--791, \href{https://arxiv.org/abs/1905.08083}{arXiv:1905.08083}.

\bibitem{Shlog}
Schnetz O., HyperlogProcedures, Version~0.5, 2021, {M}aple package available at
 \url{https://www.math.fau.de/person/oliver-schnetz/}.

\bibitem{Ysome}
Yeats K., Some combinatorial interpretations in perturbative quantum field
 theory, in Feynman Amplitudes, Periods and Motives, \textit{Contemp. Math.},
 Vol.~648, \href{https://doi.org/10.1090/conm/648/13006}{Amer. Math. Soc.}, Providence, RI, 2015, 261--289,
 \href{https://arxiv.org/abs/1302.0080}{arXiv:1302.0080}.

\bibitem{Ycirc}
Yeats K., A few {$c_2$} invariants of circulant graphs, \href{https://doi.org/10.4310/CNTP.2016.v10.n1.a3}{\textit{Commun. Number
 Theory Phys.}} \textbf{10} (2016), 63--86, \href{https://arxiv.org/abs/1507.06974}{arXiv:1507.06974}.

\bibitem{Yscompl}
Yeats K., A special case of completion invariance for the {$c_2$} invariant of
 a graph, \href{https://doi.org/10.4153/CJM-2018-006-5}{\textit{Canad.~J. Math.}} \textbf{70} (2018), 1416--1435,
 \href{https://arxiv.org/abs/1706.08857}{arXiv:1706.08857}.

\bibitem{Ystudy}
Yeats K., A study on prefixes of $c_2$ invariants, in Algebraic Combinatorics,
 Resurgence, Mould and Applications (CARMA), Vol.~2, \textit{IRMA Lectures in
 Mathematics and Theoretical Physics}, Vol.~32, \href{https://doi.org/10.4171/205-1/7}{European Mathematical Society},
 Berlin, 2020, 367--383, \href{https://arxiv.org/abs/1805.11735}{arXiv:1805.11735}.

\bibitem{ZJ}
Zinn-Justin J., Quantum field theory and critical phenomena,
 \textit{International Series of Monographs on Physics}, Vol.~77, The
 Clarendon Press, Oxford University Press, New York, 1989.

\end{thebibliography}
\end{document}